\documentclass[11pt]{article}
\usepackage[a4paper,left=17mm,top=20mm,right=17mm,bottom=25mm]{geometry}

\usepackage{authblk}
\usepackage[numbers, sort&compress]{natbib}
\usepackage[margin=30pt,font=small,labelfont=bf]{caption}
\usepackage{amsmath,amsopn,amsthm,amssymb}

\usepackage{graphicx}
\usepackage{newtxtext}
\usepackage{newtxmath}
\usepackage[mathscr]{euscript}
\usepackage{natbib}
\usepackage{mathtools}
\usepackage{enumitem}
\usepackage{xspace}
\usepackage{scalerel}
\usepackage{bold-extra}

\usepackage{tikz}
\usetikzlibrary{positioning,calc}

\usepackage[
 	colorlinks=true,
	urlcolor=black,
	linkcolor=blue,
    citecolor=blue
]{hyperref}

\newcommand{\Hasse}{\mathscr{H}}

\renewcommand{\S}{\mathcal{S}}

\DeclareMathOperator{\CC}{\mathtt{C}}

\DeclareMathOperator{\indeg}{indeg}
\DeclareMathOperator{\outdeg}{outdeg}

\newcommand{\dist}{d_{\ominus}}

\newtheorem{theorem}{Theorem}[section]
\newtheorem{lemma}[theorem]{Lemma} 
\newtheorem{corollary}[theorem]{Corollary}
\newtheorem{proposition}[theorem]{Proposition}

\newtheorem{definition}[theorem]{Definition}

\newtheorem{remark}[theorem]{Remark}

\usepackage{xcolor}

\providecommand{\keywords}[1]{\textbf{\textit{Keywords: }} #1}

\title{The $\ominus$-metric to compare phylogenetic networks}

\author[1,2]{Marc Hellmuth}

\author[3]{Manuel Lafond}

\author[4]{Guillaume E. Scholz}

\affil[1]{Department of Computer Science, Leipzig University,  DE-04109 Leipzig, Germany}

\affil[2]{Department of Mathematics, Faculty of Science,
  Stockholm University, SE-10691 Stockholm, Sweden} 

\affil[3]{Université de Sherbrooke, Canada}

\affil[4]{Institute of Mathematics and Computer Science, Greifswald University, DE-17849 Greifswald, Germany}

\date{\ }

\begin{document}
\sloppy

\maketitle

\begin{abstract}
We introduce two novel distances for comparing rooted phylogenetic networks based on the
$\ominus$-operator, which removes a vertex while preserving the ancestor relations among the
remaining vertices. The distance $\dist$ measures the minimum number of such removals needed to
obtain isomorphic networks, whereas $\dist^-$ ignores shortcut arcs and therefore compares the
induced ancestry structures. We show that $\dist$ is a metric up to leaf-fixing isomorphism and that
$\dist^-$ is a metric up to shortcut-free isomorphism. Moreover, both distances extend the
Robinson--Foulds distance on phylogenetic trees and are bounded below by the hardwired
cluster distances. For several broad network classes, including tree-child, normal, level-$1$, and
regular networks, $\dist^-$ can be computed in polynomial time. In contrast, computing $\dist$ is
NP-hard, W[2]-hard when parameterized by the distance value, and admits no polynomial-time
constant-factor approximation unless $\mathrm{P}=\mathrm{NP}$. Although computing $\dist^-$ is
NP-hard in general, for distinct-cluster networks it reduces to \textsc{Vertex Cover}, yielding a
fixed-parameter algorithm and a polynomial-time $2$-approximation.
\end{abstract}

\smallskip
\noindent
\keywords{phylogenetic networks; network comparison; operational distance;
Robinson-Foulds distance; hardwired clusters; shortcut-free networks; distinct-cluster networks;
vertex cover; fixed-parameter tractability; approximation algorithms}


\section{Introduction}

Reconstructing evolutionary histories is a central task in computational biology. Phylogenetic trees
provide a suitable model when evolution proceeds exclusively through branching events. However,
processes such as hybridization, horizontal gene transfer, and recombination may produce reticulate
patterns of ancestry that cannot be represented faithfully by a tree. Rooted phylogenetic networks
provide a more general framework for describing such evolutionary histories, and a growing number of
methods are available for reconstructing them from biological
data~\cite{solis2016inferring,wen2018inferring,lutteropp2022netrax,kong2025inference}.

The increasing number of reconstruction methods creates a corresponding need for meaningful ways to
compare their outputs. For example, one may wish to compare an inferred network with a simulated or
otherwise known reference network, compare the results produced by different reconstruction methods,
or quantify the variability among networks obtained from different data sets. This is generally more
difficult for networks than for trees. In a phylogenetic tree, features such as clusters, displayed
triplets, and ancestor relations are tightly linked and determine the tree topology
uniquely~\cite{SempleSteel2003}. For phylogenetic networks, however, these features capture
different aspects of the structure, and distinct networks may agree on some of them while differing
on others~\cite{gambette2012encodings,willson2016comparing,Hellmuth2023}. Consequently, there is no single
universally accepted distance for comparing arbitrary phylogenetic
networks~\cite{huson2010phylogenetic,Wang2019}.

Several dissimilarity measures have been proposed. One of the most widely used is the
\emph{hardwired cluster distance}, which counts the clusters that occur in one network but not in
the other~\cite{cardona2008metrics}. This generalizes the classical Robinson-Foulds distance on
phylogenetic trees~\cite{robinson1981comparison}. For general networks, however, the hardwired
cluster distance is only a pseudometric: two non-isomorphic networks may induce the same cluster set
and therefore have distance zero. Other \emph{feature-based} measures rely, for example, on
softwired clusters~\cite{lu2017program,gambette2017challenge}, $\mu$-representations
\cite{cardona2008comparison,cardona2024comparison,maxfield2025dissimilarity,reichling2026metrics}, or displayed rooted
triplets~\cite{gambette2012encodings,jansson2021computing}. In general, these measures are likewise
limited to pseudometrics. On certain restricted classes of networks, however, some of these
representations do give rise to genuine metrics.

A different approach is provided by \emph{operational} distances. Rather than comparing selected
features of two networks, an operational distance asks for a minimum number of modifications needed
to transform one network into another or to reduce both networks to a common structure. Notable
examples include distances based on rooted nearest-neighbor interchange
(rNNI)~\cite{gambette2017rearrangement,erdhos2021rooted}, tail-moves \cite{JanssenEtAl2018} and rooted subtree or subnet
prune-and-regraft operations (rSPR)~\cite{bordewich2017lost,klawitter2018subnet}. Such distances
often distinguish networks more effectively, but many are defined only for restricted classes of
networks or for networks having the same number of reticulations. Other operational measures apply,
for example, to orchard networks or LGT-networks~\cite{landry2022defining,marchand2026comparison}.
The recently introduced contraction distance can compare arbitrary networks, but it is only a
semimetric because it does not satisfy the triangle inequality~\cite{marchand2025finding}. 
Moreover, most operational distances are NP-hard to compute, even on restricted classes of
phylogenetic networks~\cite{JanssenEtAl2018}.

In summary, there is still no universally satisfactory metric for comparing phylogenetic networks.
This motivates the search for novel operational distances that apply to broad classes of networks,
distinguish non-isomorphic structures, satisfy the metric axioms, and retain a transparent
interpretation in terms of ancestry. In this paper, we introduce two such distances based on the
$\ominus$-operator. This operator was introduced in \cite{SCHS:24} and transforms a network $N$
into $N\ominus v$ by deleting a vertex $v$ together with all incident arcs and then connecting each
parent of $v$ with each child of $v$. It can therefore be viewed as a natural generalization of
suppressing vertices in rooted phylogenetic networks. The $\ominus$-operator has since been used to
simplify DAGs while preserving the ancestor relations among the retained vertices
\cite{HL:24,HELLMUTH2026584} and to relate normal networks to a regularization procedure based on
successive vertex removal~\cite{HLM:26-RegNormArxiv}. Since the order in which distinct vertices are
removed is irrelevant, the operation naturally extends to vertex sets: for $W\subseteq V(N)$, we
write $N\ominus W$ for the DAG obtained by successively applying $\ominus$ to all vertices in $W$.
These properties make the $\ominus$-operator a natural basis for operational distances between
phylogenetic networks.

\begin{figure}[t]
\centering
\includegraphics[width=0.9\textwidth]{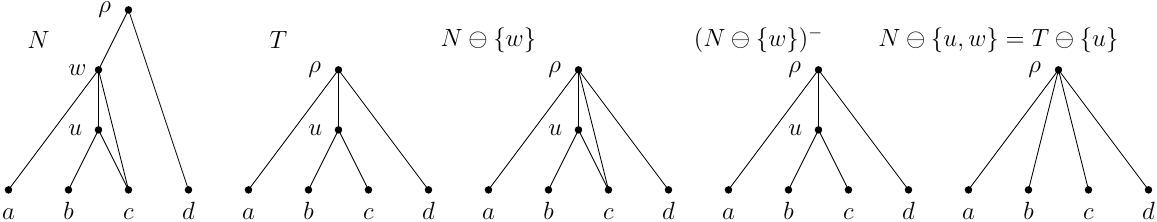}
\caption{
A network $N$, a phylogenetic tree $T$, and several $\ominus$-reductions. The arc $\rho\to c$ in
$N\ominus w$ is a shortcut, since there is also the directed path $\rho\to u\to c$. Here, $
(N\ominus w)^-\simeq T$, and therefore $\dist^-(N,T)=1$. In contrast, $N\ominus w$ is not isomorphic
to $T$. To obtain isomorphic networks without removing shortcuts, we must instead ``suppress'' both $u$
and $w$ from $N$ and  $u$ from $T$. This yields $ N\ominus\{u,w\}\simeq T\ominus u$, where
both reduced networks are star trees. Hence, $\dist(N,T)=3$.
}
\label{fig:dist-vs-dist-minus}
\end{figure}

Let $N$ and $N'$ be networks on the same leaf set. We define the $\ominus$-distance $\dist(N,N')$ as
the minimum total number $|W|+|W'|$ of internal vertices that must be removed from $N$ and $N'$,
respectively, so that the resulting networks are isomorphic: $ N\ominus W\simeq N'\ominus W'$. An
optimal pair $(W,W')$ therefore identifies a common $\ominus$-reduction with the maximum possible
total number of retained vertices. In contrast to a distance that merely counts differing features,
the resulting common reduction also indicates which vertices and ancestry relations are responsible
for the differences between the networks, see Figure~\ref{fig:dist-vs-dist-minus}
for an illustrative example.

Applying $\ominus$ may create shortcuts, that is, arcs whose endpoints are already connected by a
longer directed path. Such arcs are redundant from an ancestry perspective, since removing them does
not change the ancestor relation. For a network $N$, let $N^-$ denote the shortcut-free DAG obtained
from $N$ by removing all shortcuts. The network $N^-$ is the unique inclusion-minimal DAG on $V(N)$
that induces the same ancestor relation as $N$. Passing to shortcut-free reductions therefore
removes representation-dependent redundancy and ensures that networks encoding the same ancestry
structure are treated as equivalent. This motivates the relaxed distance $\dist^-(N,N')$, which
requires only that $ (N\ominus W)^-\simeq (N'\ominus W')^-$. Thus, $\dist$ compares the complete arc
structure of the reduced networks, whereas $\dist^-$ compares their induced ancestor relations. In
particular, $\dist^-$ is insensitive to redundant shortcut arcs that may arise during the
$\ominus$-reduction itself.
Beyond providing a numerical measure of dissimilarity, the $\ominus$-distances also produce
interpretable witnesses for the differences between two networks. An optimal pair $(W,W')$ of deletion sets
identifies vertices whose removal resolves the structural or ancestral disagreements, while the
resulting common $\ominus$-reduction represents the ancestry structure shared by both networks.
Thus, the distances may be used not only to rank or benchmark reconstructed networks, but also to
localize their differences and extract a common evolutionary core \cite{HL:24}.

We establish several structural and algorithmic properties of these distances. In
Sections~\ref{sec:prelim} and~\ref{sec:ominus}, we introduce the necessary definitions and recall
basic properties of the $\ominus$-operator. We then show in Section~\ref{sec:metric} that $\dist$ is
a metric on networks with a fixed leaf set, considered up to leaf-fixing isomorphism, and that
$\dist^-$ is a metric up to shortcut-free isomorphism. Hence, unlike many existing network
dissimilarities, both distances satisfy the triangle inequality and distinguish precisely the
structures they are intended to compare.

In Section~\ref{sec:ominus-hardwired}, we relate $\dist$ and $\dist^-$ to two cluster-based
dissimilarity measures. For a network $N$, let $\mathfrak C_N$ denote its set of clusters and let
$\mathfrak M_N$ denote the corresponding multiset, in which clusters are counted according to their
multiplicities. We show that, for arbitrary networks $N$ and $N'$ on the same leaf set, \[
d_{RF}(N,N') = |\mathfrak C_N\Delta\mathfrak C_{N'}| \leq |\mathfrak M_N\Delta\mathfrak M_{N'}| \leq
\dist^-(N,N') \leq \dist(N,N'), \] where $\Delta$ denotes the symmetric difference. Here, $d_{RF}$ 
is the aforementioned hardwired cluster distance, which extends the Robinson-Foulds distance to networks, 
while $|\mathfrak M_N\Delta\mathfrak M_{N'}|$ is its multiset version. We identify several
broad classes of networks for which some of these bounds are attained. In particular, on
phylogenetic trees, both $\dist$ and $\dist^-$ coincide with the Robinson-Foulds distance. Thus, the
$\ominus$-metrics extend the Robinson--Foulds distance from trees to broad classes of phylogenetic
networks while remaining genuine metrics on arbitrary networks. We further show that $\dist^-$ can
be computed in polynomial time for several important classes, including tree-child networks, normal
networks, binary level-$1$ networks, and semi-regular and regular networks. As a consequence, for
networks $N$ and $N'$ belonging to any of these classes, it can also be decided in polynomial time
whether $N^-\simeq (N')^-$.

Finally, we investigate the computational complexity of the two distances. The preceding results
show that $\dist^-(N,N')$ can be computed in polynomial time for several broad classes of networks.
In contrast, computing $\dist$ remains NP-hard even when restricted to some of these classes
as shown in Section~\ref{sec:dist-hard}. 
Moreover, computing $\dist$ is W[2]-hard when parameterized by the distance value and admits no
polynomial-time constant-factor approximation unless $\mathrm{P}=\mathrm{NP}$.

In Sections~\ref{sec:distMinus-VC} and~\ref{sec:NPHard-dist-minus}, we study the computational
complexity of $\dist^-$. We first consider \emph{distinct-cluster networks}, in which distinct
vertices induce distinct clusters. This is a natural and broad class that includes, for example,
all regular networks and thus all phylogenetic trees.
The key ideas are most transparent for \emph{DC-similar networks}, that is, distinct-cluster
networks with the same cluster set. For two such networks, we construct the \emph{bad ancestry
graph}: its vertices correspond to their canonically identified vertices, and an edge records a
disagreement between the ancestor relations induced by the two networks. We show that computing
$\dist^-$ is equivalent to computing a minimum vertex cover of this graph. Consequently, $\dist^-$
is fixed-parameter tractable when parameterized by its value and admits a polynomial-time
$2$-approximation. Using a reduction to the common cluster set, these results extend from
DC-similar networks to arbitrary distinct-cluster networks.
The vertex-cover formulation is also useful in practice. It allows one to apply established exact,
fixed-parameter, approximation, and integer-programming methods directly to the bad ancestry graph.
Moreover, a vertex cover translates into deletion sets for the original networks. Thus, the
computation provides not only the distance value but also vertices whose removal resolves the
disagreements between the two ancestor relations. 

We close this contribution with a short
summary and outlook in Section~\ref{sec:outlook}.


\section{Preliminaries}
\label{sec:prelim}

\paragraph{Sets and Multisets.} 
In what follows, $X$ will always denote a finite non-empty set.
A \emph{multiset} is a set-like collection in which elements may occur more than once. The
number of occurrences of an element $x$ in a multiset $A$ is called its
\emph{multiplicity} and is denoted by $m_A(x)$.
For two sets $A$ and $B$, we write $ A\Delta B \coloneqq (A\setminus B)\cup(B\setminus A) $ for
their symmetric difference. Hence $ |A\Delta B|=|A\setminus B|+|B\setminus A|$.
We use the same notation for multisets, where elements are counted with
multiplicity. Thus, if an element $x$ has multiplicity $m_A(x)$ in $A$ and
multiplicity $m_B(x)$ in $B$, then $x$ contributes
$|m_A(x)-m_B(x)|$ to $|A\Delta B|$.

\paragraph{Graphs, DAGs and networks.}
A directed graph $G=(V,E)$ consists of a non-empty vertex set $V(G)\coloneqq V$ and an arc set
$E(G)\subseteq V\times V$. 
We sometimes write for arcs $u\to v$ instead of $(u,v)$ in directed graphs $G$.

We of write $u\leadsto v$ to denote a directed $uv$-path in $G$.
If there is a $uv$-path and the arc $v\to u$ in $G$, then $G$
contains a  directed cycle.
A directed graph without directed cycles is called a \emph{directed
acyclic graph}, or \emph{DAG}.

We write $v\preceq_G w$ if and only if there is a directed $wv$-path 
with $v=w$ allowed in $G$. If one of $v\preceq_G w$ or $w\preceq_G v$ holds, then $v$ and $w$ are \emph{$\preceq_G$-comparable}
and, otherwise,  $v$ and $w$ are \emph{$\preceq_G$-incomparable}. If $v\preceq_G w$ and $v\neq w$, we write $v\prec_G w$.

Let $G$ be a DAG.
Then, $G$ is \emph{phylogenetic} if it does not contain a vertex $v$ such that $\outdeg_G(v)\coloneqq\left|\left\{u\in V \colon (v,u)\in
E(G)\right\}\right|=1$ and
$\indeg_G(v)\coloneqq\left|\left\{u\in V \colon (u,v)\in
E(G)\right\}\right|\leq 1$.

If $u\to v$ is an arc in $G$, then $v$ a \emph{child} of $u$ and $u$ a \emph{parent} of $v$. A
vertex $x$ in $G$ without children is called a \emph{leaf} of $G$. We denote by $L(G)\subseteq V(G)$
the set of leaves of $G$. If $L(G)=X$, then $G$ is a \emph{DAG on $X$}. A vertex $v\in V(G)$ that
has no parents, is called a \emph{root} of $G$, and the set of roots of $G$ is denoted by $R(G)$.
Note that $L(G)\neq \emptyset$ and $R(G)\neq \emptyset$ for all DAGs $G$ \cite{HL:24}. A vertex $v$
of $G$ is a \emph{tree-vertex} if $\indeg_G(v) \leq 1$, and a \emph{reticulation-vertex} if
$\indeg_G(v) \geq 2$. Note that a leaf could be a reticulation-vertex. We denote by $V^0(G)$ the set
of all vertices of $G$ that are neither leaves nor roots, that is, $V^0(G)=V(G) \setminus (L(G) \cup
R(G))$. A \emph{(rooted) network} $N$ is a DAG for which $|R(N)|=1$, i.e., $N$ has a unique root
$\rho\in V(N)$. A \emph{(rooted) tree} is a network that does not contain vertices
reticulation-vertices. A \emph{star tree on $X$} is the tree $T$ with vertex set $\{\rho\}\cup X$
and arc set $\{(\rho,x)\mid x\in X\}$.

An arc $e=(u,w)$ in a DAG $G$ is a
\emph{shortcut} if there is a directed $uw$-path that does not contain the arc $e$. A DAG without shortcuts is \emph{shortcut-free}. 
Throughout the paper, we will use the \emph{shortcut-free version} $G^-$ of a DAG $G$
that is obtained from $G$ by removal of all of its shortcuts (without removing the the incidend vertices).    
\begin{lemma}[{\cite[L.~2.5]{HL:24}}]\label{lem:properties-SF-G_NEW}
    Let $G$ be a DAG on $X$. Then, $G^-$ is a shortcut-free DAG on $X$. Moreover, $V(G)=V(G^-)$ and,  for all $u, v \in V(G)$, we have
    $u \preceq_{G} v$ if and only if $u \preceq_{G^-} v$. 
\end{lemma}

\paragraph{Isomorphisms.}
Two DAGs $G$ and $H $ are \emph{graph isomorphic} if there is a graph isomorphism between $G$ and
$H$, i.e., a bijective map $\varphi\colon V(G)\to V(H)$ such that $(u,v)\in E(G)$ if and only if
$(\varphi(u),\varphi(v))\in E(H)$. In this case, we write $G\approx H$. Moreover, if $G$ and $H$ are
DAGs that also have the same leaf set $X$, then $G$ and $H$ are \emph{isomorphic}, in
symbols $G\simeq H$, if there is a graph isomorphism $\varphi$ between $G$ and $H$ that satisfies
$\varphi(x) = x$ for all $x\in X$.

Let $N,N'$ be two networks.
For $V \subseteq V(N)$ and $V' \subseteq V(N')$, we say that a map $\varphi: V \to V'$ is \emph{ancestor-preserving} if $\varphi$ is a bijective map, 
and for all $u,v \in V(N) \setminus W$, 
$u \preceq_N v$ if and only if $\varphi(u) \preceq_{N'} \varphi(v)$. 
The following simple result shows that ancestor-preserving maps on shortcut-free networks are graph isomorphisms.

\begin{lemma}\label{lm:isom}
Let $N, N'$ be two shortcut-free networks. If there exists an ancestor-preserving map $\varphi:V(N) \to V(N')$, then $\varphi$ is a graph isomorphism.
\end{lemma}
\begin{proof}
Since $\varphi: V \to V'$ is a bijective map, we can without loss of generality assume that $V(N)=V(N')$ in such a way that $\varphi(v)=v$ for all $v\in V(N)$. 
The proof follows now the same argument as used in the proof of Lemma~7.2 in \cite{willson2016comparing}.
\end{proof}

\paragraph{Clusters and related concepts.}

For a DAG $G$ on $X$ and a vertex $v\in V(G)$, the \emph{cluster} of $v$,
denoted $\CC_G(v)$, is the set of all leaves $x\in X$ such that
$x\preceq_G v$. The \emph{cluster set} of $G$ is
$ \mathfrak C_G \coloneqq \{\CC_G(v)\mid v\in V(G)\}$.
The \emph{cluster multiset} of $G$, denoted $\mathfrak M_G$, is the
multiset in which each cluster $C\in\mathfrak C_G$ occurs with multiplicity
$|\{v\in V(G)\mid \CC_G(v)=C\}|$.

Conversely, the Hasse diagram $\Hasse(\mathfrak{C})$ provides a natural way to associate a DAG to a set system $\mathfrak{C}$. 
More precisely,  $\Hasse(\mathfrak{C})$ is the DAG with vertex set
$\mathfrak{C}$ and arcs $(A,B)$ if (i) $B\subsetneq A$ and (ii) there is
no $C\in \mathfrak{C}$ with $B\subsetneq C\subsetneq A$. 
This, in turn, gives rise to the following definition
\begin{definition}[{\cite{BSS:04}}]
  \label{def:regular-N}
  A DAG $G=(V,E)$  is \emph{regular} if the map
  $\varphi\colon V\to V(\Hasse(\mathfrak{C}_G))$ defined by  $v\mapsto \CC_G(v)$ is a graph
 isomorphism between $G$ and $\Hasse(\mathfrak{C}_G)$, i.e., in symbols, $G \approx \Hasse(\mathfrak{C}_G)$.
\end{definition}
Note that, by definition,  $\Hasse(\mathfrak{C}_G)$ and, therefore, 
regular DAGs are shortcut-free. Moreover, one easily verifies that, for  
regular DAGs $G$ and $H$  it holds that   $\mathfrak C_G=\mathfrak C_H$ implies 
$ G\simeq H$.

Following \cite{willson2016comparing}, a network $N$ is \emph{distinct-cluster (DC)} if, for all $u, v
\in V(N)$, it holds that $C_N(u) = C_N(v)$ if and only if $u = v$. Thus, regular networks are
distinct-cluster. Note that distinct-cluster networks have also a close connection to so-called
LCA-relevant networks, see \cite[Thm 4.4.]{HL:24}. 

Two distinct-cluster networks $N, N'$ are \emph{DC-similar} if $\mathfrak C(N) =
\mathfrak C(N')$.  A direct consequence of this definition is that  
DC-similar networks $N$ and $N'$ must be defined on the same leaf set. 
Moreover, for DC-similar networks, there is a canonical bijection between their vertex sets: each vertex $v\in
V(N)$ is mapped to the unique vertex $v'\in V(N')$ satisfying $\CC_N(v)=\CC_{N'}(v')$ which leads to the following
\begin{remark}\label{rem:similar}
Throughout this paper, whenever $N$ and $N'$ are DC-similar, we can identify their vertex sets via
this canonical bijection. Thus we may write $V(N)=V(N')$ and assume that $\CC_N(v)=\CC_{N'}(v)$ for
every vertex $v\in V(N)$ In this case, we say that $N$ and $N'$ have \emph{canonically identified
vertex sets}.
\end{remark}

\paragraph{Further classes of networks.}
The following list provides several networks classes that have attracted attention in the
past years. To keep the paper focused, we only present here the formal definitions. For more details
on these classes, their properties, and their relevance in the field of phylogenetics, we invite the
interested reader to have a look at the references indicated.

\smallskip\noindent
A network $N$ on $X$ \dots
\begin{itemize}[noitemsep,nolistsep]
\item[\dots] satisfies  \emph{path-cluster-comparability (PCC)} if for all $u,v \in V(N)$, $\CC_N(v) \subseteq \CC_N(u)$ implies that $u$ and $v$ are $\preceq_N$-comparable \cite{Hellmuth2023}.
\item[\dots] is \emph{binary} if every tree vertex $v$ is
								either a leaf or has $\outdeg_N (v) = 2$, and every retriculation vertex
								$v$ satisfies $\indeg_N (v) = 2$ and $\outdeg_N(v) = 1$.
\item[\dots] is \emph{tree-child} if for all non-leaf vertices $v$ of $N$, there is a child $v'$ of $v$ that is a tree-vertex \cite{CRV:07}.
\item[\dots] is \emph{normal} if $N$ is tree-child and shortcut-free \cite{Willson2010}.
\item[\dots] is \emph{semi-regular} if shortcut-free and satisfies (PCC) \cite{Hellmuth2023}. 

						Note that, by \cite[Thm~2]{Hellmuth2023}, regular networks are precisely the 
						semi-regular networks that have no vertex with outdegree 1.

\item[\dots] is level-1, if each inclusion-maximal biconnected subgraph of $N$ contains at most one 
						reticulation-vertex distinct from its $\preceq_N$-maximal vertex \cite{Hellmuth2023}.
\end{itemize}


\section{The $\ominus$-operator}
\label{sec:ominus}

The $\ominus$-operator was introduced in \cite{HL:24,SCHS:24} as a generalization  of ``suppressing'' vertices in a network, while preserving key structural properties of that network.

\begin{definition}[{\cite{HL:24}}]\label{def:ominus}
  Let $N=(V,E)$ be a network and $v\in V$. Then $N\ominus v=(V',E')$ is the
  directed graph with vertex set $V'=V\setminus\{v\}$ and arcs $(p,q)\in E'$ 
  precisely if $v\ne p$,
  $v\ne q$ and $(p,q)\in E$, or if $(p,v)\in E$ and $(v,q)\in E$. 
\end{definition} 

Intuitively, applying $N\ominus v$ removes the
vertex $v$ and reconnects each parent of $v$ directly to each child of
$v$. In case $v$ is a leaf or a root,  $v$ and its incident edges are just deleted, 
see Figure~\ref{fig:ominus-simple} for an illustrative example.

The $\ominus$-operation is order-independent on sets of vertices, i.e., 
for distinct vertices $v,w\in V(N)$,
$(N\ominus v)\ominus w=(N\ominus w)\ominus v$ (cf.\ \cite[L~2.3]{HLM:26-RegNormArxiv}).
Hence, for a given set non-empty subset $W = \{w_1,\dots,w_\ell\} \subsetneq V(N)$, 
we can, without loss of generality, define 
  \[N\ominus W \coloneqq (\dots ((N \ominus w_1) \ominus w_2) \dots)\ominus w_\ell.\]
For notational reasons, put $N\ominus\emptyset=N$.

To recall, in a network $N$ on $X$ and with root $\rho$ we have $V^0(N) =V(N)\setminus (X\cup\{\rho\})$.
As shown next, the operator $\ominus$ can be used to transform any network into an
another one preserving key structural properties of the original network.

\begin{lemma}\label{lem:network-ominus}
Let $N$ be a network on $X$ with root $\rho$ and let $W\subseteq V^0(N)$. 
Then, $N'\coloneqq N\ominus W$ is a network on $X$ with root $\rho$ such that 
$u \preceq_N w$ if and only if $u \preceq_{N'} w$ for all $u, w \in V(N')$.
In particular, $\CC_N(u)=\CC_{N'}(u)$ for all $u\in V(N')$.
\end{lemma}
\begin{proof}
Let $N$ be a network on $X$ with root $\rho$, and let $W\subseteq V^0(N)$.
It has been argued in \cite[Obs~5.3]{HL:24} that the statement holds for DAGs, 
i.e.,  $N'\coloneqq N\ominus W$ is a DAG on $X$ such that 
$u \preceq_N w$ if and only if $u \preceq_{N'} w$ for all $u, w \in V(N')$.
In particular, $\CC_N(u)=\CC_{N'}(u)$ for all $u\in V(N')$.
It remains to show that $N'$ is a network. 
Since $N$ is a network with root $\rho$, every vertex $w\in V(N')$ satisfies $ w\preceq_N
\rho$.  By preservation of the ancestor relation, this implies $ w\preceq_{N'} \rho $ for every $w\in
V(N')$. Thus $\rho$ is the unique root of $N'$.
\end{proof}

\begin{figure}
	\centering
	\includegraphics[width = 0.9\textwidth]{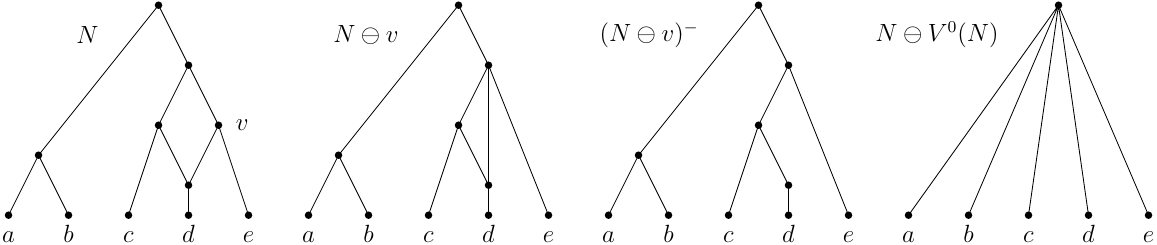}
	\caption{A network $N$, the $\ominus$-reductions $N\ominus v$, the shortcut-free $\ominus$-reduction $(N \ominus v)^-$, and the star-tree $N\ominus V^0(N)$.
	}
\label{fig:ominus-simple}
\end{figure}

The following result shows that, when computing $N \ominus W$, shortcuts can be removed at any stage of the sequence of $\ominus$-operations.

\begin{lemma}\label{lm:sfsf}
Let $N$ be a network and let $W\subseteq V(N)$. Then
$   (N\ominus W)^-  =  (N^-\ominus W)^-$.
\end{lemma}
\begin{proof}
Both $(N\ominus W)^-$ and $(N^-\ominus W)^-$ have vertex set
$V(N)\setminus W$. 
Moreover, removing shortcuts does not change the ancestor relation. Hence $N$ and
$N^-$ have the same ancestor relation, i.e., $\preceq_N = \preceq_{N^-}$. Moreover, by
Lemma~\ref{lem:network-ominus}, $N\ominus W$
preserves the ancestor relation of $N$ among the vertices in
$V(N)\setminus W$, and $N^-\ominus W$ preserves the ancestor relation of
$N^-$ among the same vertices. Therefore, $N\ominus W$ and
$N^-\ominus W$ have the same ancestor relation on $V(N)\setminus W$.

It remains to observe that, in a shortcut-free DAG, the arc set is uniquely determined by the ancestor
relation. More precisely, for vertices $u$ and $v$, the arc $(u,v)$ is present if and only if
$v\prec u$ and there is no vertex $z$ with $v\prec z\prec u$. Indeed, if such a vertex $z$
existed, then  $(u,v)$ would be a shortcut. Conversely, if $v\prec u$ and no such vertex $z$
exists, then every directed path from $u$ to $v$ has no internal vertex and therefore must consist
of the single arc $(u,v)$. Since $N\ominus W$ and $N^-\ominus W$ have the same ancestor relation on
the same vertex set, their shortcut-free reductions have the same vertex
set and the same arcs. Hence, $ (N\ominus W)^-= (N^-\ominus W)^-$,
\end{proof}

The following simple result will be used throughout this paper. 

\begin{lemma}\label{lm:stt}
 	For all networks $N$ on $X$, $N\ominus V^0(N)$ is a star-tree on $X$.
\end{lemma}
\begin{proof} 
 	 Let $N$ be a network on $X$ and let $\rho$ denote the unique root of $N$.
By definition, $V(N\ominus V^0(N)) = X\cup \{\rho\}$. 
	By Lemma~\ref{lem:network-ominus}, $N\ominus V^0(N)$ is a network on $X$
	with root $\rho$. 	Since, $\rho$ is the only non-leaf vertex and since the leaf
	set of $N\ominus V^0(N)$ is $X$ it immediately follows that 
	$N\ominus V^0(N)$ is a star-tree on $X$.
\end{proof}

The following lemma, which will be useful later, shows that the $\ominus$-operator can be used to
remove precisely those vertices whose clusters occur in only one of two distinct-cluster networks.
In particular, it reduces the comparison of two distinct-cluster networks to the comparison of a
pair of DC-similar networks.

\begin{lemma}\label{lem:dc-to-dc-similar}
Let $N$ and $N'$ be distinct-cluster networks on the same leaf set $X$, and put $ D_N\coloneqq
\{v\in V^0(N):\CC_N(v)\notin\mathfrak C_{N'}\} $ and $ D_{N'}\coloneqq \{v'\in
V^0(N'):\CC_{N'}(v')\notin\mathfrak C_N\}$. Then $N\ominus D_N$ and $N'\ominus D_{N'}$ are
DC-similar and it holds that
\[ \mathfrak C_{N\ominus D_N} = \mathfrak C_N\cap\mathfrak C_{N'} = \mathfrak C_{N'\ominus D_{N'}}
\quad \text{and} \quad |D_N|+|D_{N'}| = |\mathfrak C_N\mathbin{\Delta}\mathfrak C_{N'}|.
 \] 
\end{lemma}
\begin{proof}
By Lemma~\ref{lem:network-ominus}, applying the $\ominus$-operator to vertices in $D_N$ preserves
the clusters of all remaining vertices. Hence, $ \mathfrak C_{N\ominus D_N} = \{\CC_N(v):v\in
V(N)\setminus D_N\}. $ We first show that $ \mathfrak C_{N\ominus D_N} = \mathfrak C_N\cap\mathfrak
C_{N'}. $

Let $C\in\mathfrak C_{N\ominus D_N}$. Then $C=\CC_N(v)$ for some $v\in V(N)\setminus D_N$. If $v\in
V^0(N)$, then $v\notin D_N$ implies $\CC_N(v)\in\mathfrak C_{N'}$. If $v$ is a leaf, then $C=\{x\}$
for some $x\in X$, and this cluster also occurs in $N'$. If $v$ is the root of $N$, then $C=X$,
which is also the cluster of the root of $N'$. Thus $ C\in\mathfrak C_N\cap\mathfrak C_{N'}. $
Consequently, $ \mathfrak C_{N\ominus D_N} \subseteq \mathfrak C_N\cap\mathfrak C_{N'}. $

Conversely, let $C\in\mathfrak C_N\cap\mathfrak C_{N'}$. Since $N$ is distinct-cluster, there is a
unique vertex $v\in V(N)$ satisfying $\CC_N(v)=C$. If $v\in V^0(N)$, then $C\in\mathfrak C_{N'}$
implies $v\notin D_N$. Roots and leaves do not belong to $D_N$ by definition. Hence, in every case,
$v\in V(N)\setminus D_N$. Lemma~\ref{lem:network-ominus} now implies $ C=\CC_N(v)=\CC_{N\ominus
D_N}(v), $ and therefore $C\in\mathfrak C_{N\ominus D_N}$. Thus $ \mathfrak C_N\cap\mathfrak C_{N'}
\subseteq \mathfrak C_{N\ominus D_N}. $ We conclude that $ \mathfrak C_{N\ominus D_N} = \mathfrak
C_N\cap\mathfrak C_{N'}. $ By similar arguments, $ \mathfrak C_{N'\ominus D_{N'}} = \mathfrak
C_N\cap\mathfrak C_{N'}. $

Now we show that the two $\ominus$-reduced networks are DC-similar. Let $u,v\in V(N)\setminus D_N$
and suppose that $ \CC_{N\ominus D_N}(u)=\CC_{N\ominus D_N}(v). $ Again by
Lemma~\ref{lem:network-ominus}, $ \CC_N(u)=\CC_N(v). $ Since $N$ is distinct-cluster, this implies
$u=v$. Hence $N\ominus D_N$ is distinct-cluster. The same argument shows that $N'\ominus D_{N'}$ is
distinct-cluster. Since the two reduced networks have the same cluster set, they are DC-similar.

Finally, because $N$ is distinct-cluster, the map $ v\mapsto \CC_N(v) $ is a bijection from $V(N)$
to $\mathfrak C_N$. Moreover, every cluster in $\mathfrak C_N\setminus\mathfrak C_{N'}$ is
represented by an internal vertex of $N$: the root cluster $X$ and all singleton leaf clusters occur
in both networks. Therefore, $ |D_N| = |\mathfrak C_N\setminus\mathfrak C_{N'}| $. Similarly, $
|D_{N'}| = |\mathfrak C_{N'}\setminus\mathfrak C_N| $. Consequently, $ |D_N|+|D_{N'}| = |\mathfrak
C_N\setminus\mathfrak C_{N'}| + |\mathfrak C_{N'}\setminus\mathfrak C_N|\ = |\mathfrak
C_N\mathbin{\Delta}\mathfrak C_{N'}|$.
\end{proof}


\section{The $\ominus$-metrics $\dist$ and $\dist^-$}
\label{sec:metric}

The aim of this section is to define a metric based on the $\ominus$-operator.
To this end, we define first a distance between two networks based on this operator as follows. 

\begin{definition}
Let $N$ and $N'$ be networks on $X$. 
Put
\[ \mathtt W(N,N') \coloneqq \{(W,W')\mid W\subseteq V^0(N), W'\subseteq V^0(N') \text{ and }
N\ominus W\simeq N'\ominus W'\}.     \] 

Then, the \emph{$\ominus$-distance} between $N$ and $N'$ is defined as 
\[\dist(N,N') \coloneqq \min_{(W,W') \in \mathtt W(N,N')}  \{|W|+|W'|\}. \]
\end{definition}

Observe first that whenever, $(W,W')\in \mathtt W(N,N')$ then since, $W\subseteq V^0(N)$ and $W'\subseteq V^0(N')$, 
Lemma~\ref{lem:network-ominus} implies that the two isomorphic graphs $N\ominus W$ and $N'\ominus W'$
are network on $X$.
Thus,  distance $\dist(N,N')$ is the minimum total number of vertices that have to be removed, via the
    $\ominus$-operator, from $N$ and $N'$ in order to obtain isomorphic networks on $X$.
    The next simple consequence of Lemma~\ref{lm:stt} shows that $\dist(N,N')$ is well-defined, as we can always reduce two networks on $X$ to star trees.

\begin{corollary}\label{cor:dist-well-defined}
Let $N$ and $N'$ be networks on the same leaf set $X$. Then
$\dist(N,N')$ is well-defined.
\end{corollary}
\begin{proof}
Let $N$ and $N'$ be networks on the same leaf set $X$. 
By Lemma~\ref{lm:stt}, $N\ominus V^0(N)$ 
and $N'\ominus V^0(N')$ are both star trees on $X$
and it readily follows that $N\ominus V^0(N)\simeq N'\ominus V^0(N')$. 
Hence,  $(V^0(N),V^0(N'))\in\mathtt W(N,N')$ and thus, $W(N,N')\neq \emptyset$.
Since $V^0(N)$ and $V^0(N')$ are finite, there are only finitely many choices for $W\subseteq V^0(N)$ and $W'\subseteq
V^0(N')$. Thus $\mathtt W(N,N')$ is a non-empty finite set. Consequently, the minimum $
\min_{(W,W')\in\mathtt W(N,N')} (|W|+|W'|)$ exists, and so $\dist(N,N')$ is well-defined.
\end{proof}

We are now in the position to show that $\dist$ is a metric. Since
$\dist(N,N')=0$ holds precisely when $N$ and $N'$ are 
isomorphic, this metric is naturally defined on networks considered up to
isomorphism. Equivalently, in the identity axiom below,
equality "=" of networks is understood as $\simeq$.

\begin{theorem}\label{thm:distmetric}
Let $\mathcal N_X$ denote the class of networks on $X$. Then $\dist$ is a metric on $\mathcal N_X$
up to $\simeq$.
 More precisely,
for all networks $N,N',N''\in \mathcal N_X$, the following hold:
\begin{enumerate}[nolistsep]
    \item $\dist(N,N')\geq 0$. 
    \item $\dist(N,N')=0$ if and only if $N\simeq N'$.
    \item $\dist(N,N')=\dist(N',N)$.
    \item $\dist(N,N')\leq \dist(N,N'')+\dist(N'',N')$.
\end{enumerate} 
\end{theorem}
\begin{proof}
	Condition (1) is clear, since $\dist$ is defined as a minimum of
	cardinalities. Condition (3) follows immediately from the fact
	that $\simeq$ is symmetric.

	For Condition (2), suppose first that $\dist(N,N')=0$. Then there are
	$(W,W')\in\mathtt W(N,N')$ with $|W|+|W'|=0$. Hence
	$W=W'=\emptyset$, and so $N\simeq N'$. Conversely, if $N\simeq N'$, then
	$(\emptyset,\emptyset)\in\mathtt W(N,N')$, and therefore
	$\dist(N,N')=0$.

	We show now that Condition (4) holds. 	In the following, we define for a map $f\colon M \to M'$
	and a subset $A\subseteq M$ the map
	$f(A) \coloneqq \{f(a) \colon a\in A\}$. Moreover, we put $f(\emptyset)=\emptyset$. 
	Note that, if $f$ is a bijection, then 
	\begin{equation}\label{eq:f}
		|f(A)|=|A| \text{ for any subset } A\subseteq M. 
	\end{equation}	
	We use the latter arguments for $f=\varphi$ or $f=\varphi'$ below.
	
	Now, let $N, N', N''$ be networks on $X$. 
	Let $W\subseteq V^0(N)$ and $W_1\subseteq V^0(N'')$
	 be two sets satisfying
	$N\ominus W \simeq N''\ominus W_1$ and $\dist(N,N'') = |W| + |W_1|$. 
	Let \[\varphi \colon V(N'')\setminus W_1 \to V(N)\setminus W\] be an isomorphism between
	$ N''\ominus W_1$ and $N\ominus W $. 	
	Similarily, let $W'\subseteq V^0(N')$ and $W_2\subseteq V^0(N'')$  be two sets satisfying
	$N'\ominus W' \simeq N''\ominus W_2$ and $\dist(N',N'') = |W'| + |W_2|$. 
	Let 
	\[\varphi' \colon V(N'')\setminus W_2 \to V(N')\setminus W'\]
	be an isomorphism between $ N''\ominus W_2$ and $N'\ominus W'$. 
	Since $N''\ominus W_1  \simeq N\ominus W$ and $\varphi$ is an isomorphism between the two graphs
	it follows that 
	\[(N''\ominus W_1) \ominus v \simeq (N\ominus W) \ominus \varphi(v)\]
	for all $v\in V(N'')\setminus W_1$.
	Since $W_2\setminus W_1\subseteq V^0(N'')\setminus W_1$, we can conclude that $\varphi(W_2\setminus W_1)$ is well-defined. 
	The latter two arguments imply that 
	\begin{equation}\label{eq:a}
	 (N''\ominus W_1) \ominus (W_2\setminus W_1) \simeq (N\ominus W) \ominus \varphi(W_2\setminus W_1).
	 \end{equation}
	By analogous argumentation, we obtain 
	\begin{equation}\label{eq:b}
	(N''\ominus W_2) \ominus (W_1\setminus W_2) \simeq (N'\ominus W') \ominus \varphi'(W_1\setminus W_2).
	 \end{equation}
	Since $(N''\ominus W_1) \ominus (W_2\setminus W_1) =  N'' \ominus (W_1\cup W_2) = (N''\ominus W_2) \ominus (W_1\setminus W_2)$
	it follows that the RHS of Eq.~\eqref{eq:a} and \eqref{eq:b} satisfy
	\begin{equation}\label{eq:1}
	  (N\ominus W) \ominus \varphi(W_2\setminus W_1)  \simeq  (N'\ominus W') \ominus \varphi'(W_1\setminus W_2) .
	 \end{equation}	 
	 Note that $W\cap \varphi(W_2\setminus W_1) = \emptyset$ since all vertices $\varphi(v)$
	 in the image of $\varphi$ satisfy   $\varphi(v)\in V(N)\setminus W$.
	 Thus, ${|W\cup \varphi(W_2\setminus W_1)|  =|W| + |\varphi(W_2\setminus W_1)|}$. 
	 Similarly, $|W'\cup \varphi'(W_1\setminus W_2)| = |W'|+ |\varphi'(W_1\setminus W_2)|$. 
	 
	 	Moreover, since $\varphi$ and $\varphi'$ are isomorphisms between networks on $X$, they map
	 	    roots to roots and fix all leaves. Hence $\varphi(W_2\setminus W_1)\subseteq V^0(N)$ and
	 	    $\varphi'(W_1\setminus W_2)\subseteq V^0(N')$. Thus the sets $ W\cup\varphi(W_2\setminus
	 	    W_1) $ and $ W'\cup\varphi'(W_1\setminus W_2) $ are admissible deletion sets
	 	    for $\dist(N,N')$.
In particular, Eq.~\ref{eq:1} can be rewritten as 
	\begin{equation}\label{eq:2}
	  N\ominus (W \cup\varphi(W_2\setminus W_1))  \simeq  N'\ominus (W' \cup \varphi'(W_1\setminus W_2)) .
	 \end{equation}
	 Taking the latter arguments together we obtain
	\begin{align*} 
		\dist(N,N') & \leq |W\cup \varphi(W_2\setminus W_1)| + |W'\cup \varphi'(W_1\setminus W_2)|\\
						& = |W| + |\varphi(W_2\setminus W_1)| + |W'|+ |\varphi'(W_1\setminus W_2)| 	 \\
		            & = |W| + |W_2\setminus W_1| +  |W'|+ |W_1\setminus W_2| \quad \text{(by Eq.~\ref{eq:f})} \\ 
		            & \leq |W| + |W_2|  +  |W'| + |W_1| \\
		            & = \dist(N,N'') + \dist(N',N'')
		             = \dist(N,N'') + \dist(N'',N')
	\end{align*}
Thus Condition (4) holds.
\end{proof}

\begin{figure}
	\centering
	\includegraphics[width = 0.9\textwidth]{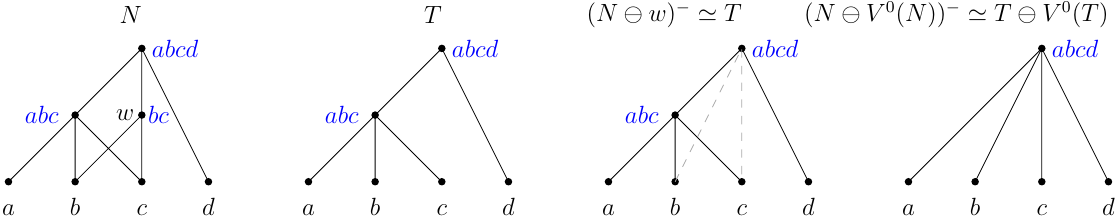}
	\caption{A network $N$ and a phylogenetic tree $T$. The cluster of non-leaf vertices are
	indicated next to the respective vertex. Both $N$ and $T$ are distinct-cluster networks.
	In addition, the networks $(N\ominus w)^- \simeq T$ (where shortcuts of $ N\ominus w$ are highlighted by
	gray-dashed arcs) and the star-tree $N\ominus V^0(N) \simeq T\ominus V^0(T)$ are shown. 
	Here, $
	|\mathfrak C_N \Delta \mathfrak C_{N'}| = |\mathfrak M_N \Delta \mathfrak M_{N'}| = |\{\{b,c\}\}| =1=	\dist^-(T,N) = |\{w\}|+|\emptyset|
	< \dist(T,N)=|V^0(N)|+|V^0(T)|=3$.}
\label{fig:TN1}
\end{figure}

As pointed out in \cite{HL:24}, the $\ominus$-operator may create
shortcuts, see for example the network $N\ominus w$
in Figure~\ref{fig:TN1}.
  Since shortcuts do not change the ancestor relation, it is
natural to ignore them when comparing two $\ominus$-reductions. We
therefore consider the following relaxed version of the distance.

For networks $N$ and $N'$ on $X$, put
\[ \mathtt W^-(N,N') \coloneqq \{(W,W')\mid W\subseteq V^0(N),\ W'\subseteq V^0(N') \text{ and } (N\ominus W)^-\simeq (N'\ominus W')^-\}. \]
Then define
\[
    \dist^-(N,N')
    \coloneqq
    \min_{(W,W')\in\mathtt W^-(N,N')}
    \{|W|+|W'|\}.
\]

Since Lemma~\ref{lm:stt} implies that, for all networks $N$ on $X$, 
$N\ominus V^0(N)$ is a star-tree on $X$ we can conclude, 
by the same arguments as used in the proof of Cor~\ref{cor:dist-well-defined}, 
that $\dist^-(N,N')$ is well-defined.  
Note also that, by Lemma~\ref{lm:sfsf}, $\dist^-(N,N')= \dist^-(N^-,N'^-)$ holds for all networks $N,N'$ on $X$.

\begin{theorem}\label{thm:dist-minus-metric}
Let $N\simeq^- N'$ precisely if $N^-\simeq (N')^-$ for any two networks $N$ and $N'$. 
Then, the distance $\dist^-$ is a metric on networks on $X$ up to
$\simeq^-$. More precisely, for all networks $N,N',N''$ on $X$,
the following hold:
\begin{enumerate}[nolistsep]
    \item $\dist^-(N,N')\geq 0$.
    \item $\dist^-(N,N')=0$ if and only if $N\simeq^- N'$.
    \item $\dist^-(N,N')=\dist^-(N',N)$.
    \item $\dist^-(N,N')\leq \dist^-(N,N'')+\dist^-(N'',N')$.
\end{enumerate}
\end{theorem}
\begin{proof}
Conditions (1) and (3) follow immediately from the definition of
$\dist^-$. For (2), note that $\dist^-(N,N')=0$ holds if and only if
$(\emptyset,\emptyset)\in\mathtt W^-(N,N')$, that is, if and only if
$   (N\ominus\emptyset)^-\simeq (N'\ominus\emptyset)^-.$
Equivalently, $N^-\simeq (N')^-$, i.e., $N\simeq^- N'$.

The triangle inequality is proved by the same argument as for Theorem~\ref{thm:distmetric}. The only
difference is that, instead of requiring a common $\ominus$-reduction $ N\ominus W \simeq N'\ominus
W'$, we require the corresponding shortcut-free reductions to be isomorphic, that is, $ (N\ominus
W)^- \simeq (N'\ominus W')^-$ .
Thus, throughout the proof of Theorem~\ref{thm:distmetric}, the admissible set $\mathtt W$ is
replaced by $\mathtt W^-$, and every network obtained by $\ominus$-editing is replaced by its
shortcut-free reduction. Lemma~\ref{lm:sfsf} ensures that these shortcut-removal
steps are compatible with the additional $\ominus$-edits used in the triangle argument; more
precisely, whenever $W\subseteq V^0(N)$ and $U\subseteq V^0(N)\setminus W$, we have $ ((N\ominus
W)^-\ominus U)^- \simeq (N\ominus (W\cup U))^-$ .
Hence the same construction as in the proof of  Theorem~\ref{thm:distmetric} yields
$ \dist^-(N,N') \leq \dist^-(N,N'')+\dist^-(N'',N')$.
Hence Condition (4) holds.
\end{proof}

Lemma~\ref{lem:dc-to-dc-similar} shows that, by removing all vertices whose clusters occur in only
one of two distinct-cluster networks, one obtains a pair of DC-similar networks. The next
proposition describes precisely how this reduction affects the $\ominus$-distances. The removed vertices
contribute exactly $|\mathfrak C_N\Delta\mathfrak C_{N'}|$, while the remaining contribution is the
corresponding $\ominus$-distance between the two reduced networks. Consequently, the 
$\ominus$-distances between
arbitrary distinct-cluster networks are completely determined by the symmetric difference of their
cluster sets and the distances between the associated DC-similar networks. This allows us to focus
primarily on the DC-similar case in the remainder of the paper.

\begin{proposition}\label{prop:dc-to-dc-similar}
Let $N$ and $N'$ be distinct-cluster networks on $X$. Let $ D_N \coloneqq \{v\in V^0(N)\mid \CC_N(v)\notin
\mathfrak C_{N'}\} $ and $ D_{N'} \coloneqq \{v'\in V^0(N')\mid \CC_{N'}(v')\notin \mathfrak C_N\}$.
Then, 
\[\dist^-(N,N') = |\mathfrak C_N\Delta \mathfrak C_{N'}| + \dist^-(N\ominus D_N, N'\ominus D_{N'}) = |D_N|+|D_{N'}| + \dist^-(N\ominus D_N,N'\ominus D_{N'}) ,\]
\[\dist(N,N') = |\mathfrak C_N\Delta \mathfrak C_{N'}| + \dist(N\ominus D_N, N'\ominus D_{N'}) = |D_N|+|D_{N'}| + \dist(N\ominus D_N,N'\ominus D_{N'}).\]
\end{proposition}
\begin{proof}
Put $ \widehat N\coloneqq N\ominus D_N$ and $\widehat N'\coloneqq N'\ominus D_{N'}$. 
We start to prove
the statement for $\dist$ and focus first on the stated equalities involving $|D_N|+|D_{N'}|$.

We first  Let $(W,W')\in\mathtt W(N,N')$. Thus, $N\ominus W\simeq N'\ominus W'$. 
We claim that $ D_N\subseteq W$ and $D_{N'}\subseteq W'$.
Suppose, for contradiction, that there is some $v\in D_N\setminus W$. By
Lemma~\ref{lem:network-ominus}, the vertex $v$ remains in $N\ominus W$ with cluster $ \CC_{N\ominus
W}(v)=\CC_N(v) $. Let $ \varphi\colon N\ominus W\to N'\ominus W' $ be an isomorphism.
Since $\varphi$ fixes all leaves, it preserves clusters. Hence $ \CC_N(v) = \CC_{N\ominus W}(v) =
\CC_{N'\ominus W'}(\varphi(v)) = \CC_{N'}(\varphi(v)). $ Thus $\CC_N(v)\in\mathfrak C_{N'}$,
contradicting $v\in D_N$. Therefore $D_N\subseteq W$. By similar arguments, $D_{N'}\subseteq W'$.

Since the $\ominus$-operation is order-independent, $ N\ominus W = (N\ominus D_N)\ominus(W\setminus
D_N) = \widehat N\ominus(W\setminus D_N), $ and similarly $ N'\ominus W' = \widehat
N'\ominus(W'\setminus D_{N'}). $ It follows that $ (W\setminus D_N,W'\setminus D_{N'}) \in \mathtt
W(\widehat N,\widehat N')$. Therefore, $|W\setminus D_N|+|W'\setminus D_{N'}|
\geq \dist(\widehat N,\widehat N')$.
Consequently, \[ |W|+|W'| = |D_N|+|D_{N'}| + |W\setminus
D_N|+|W'\setminus D_{N'}|\ \geq |D_N|+|D_{N'}| + \dist(\widehat N,\widehat N'). \] As the latter
holds for all $(W,W')\in\mathtt W(N,N')$, we can conclude that $ \dist(N,N') \geq |D_N|+|D_{N'}| +
\dist(\widehat N,\widehat N')$.

For the converse, let $U\subseteq V^0(\widehat N)$ and $U'\subseteq V^0(\widehat N')$ be chosen such
that $\widehat N\ominus U \simeq \widehat N'\ominus U'$ and $\dist(\widehat N,\widehat
N')=|U|+|U'|$. Since only internal vertices were removed when constructing $\widehat N$ and
$\widehat N'$, Lemma~\ref{lem:network-ominus} implies that $V^0(\widehat N)=V^0(N)\setminus D_N$ and
$V^0(\widehat N')=V^0(N')\setminus D_{N'}$. Hence $U\subseteq V^0(N)\setminus D_N$ and $U'\subseteq
V^0(N')\setminus D_{N'}$. Order-independence of $\ominus$ gives $ N\ominus(D_N\cup U) = \widehat
N\ominus U \simeq \widehat N'\ominus U' = N'\ominus(D_{N'}\cup U') $. Hence $ (D_N\cup U,D_{N'}\cup
U')\in\mathtt W(N,N'), $ and therefore
\[ \dist(N,N') \leq |D_N\cup U|+|D_{N'}\cup U'|\ = |D_N|+|D_{N'}|+|U|+|U'|\ = |D_N|+|D_{N'}| + \dist(\widehat N,\widehat N'). \]
In summary, we have shown that $ \dist(N,N') = |D_N|+|D_{N'}| + \dist(\widehat N,\widehat N'). $
Lemma~\ref{lem:dc-to-dc-similar} now yields $ \dist(N,N') = |\mathfrak C_N\mathbin{\Delta}\mathfrak
C_{N'}| + \dist(N\ominus D_N,N'\ominus D_{N'})$.

The proof for $\dist^-$ is analogous. Let $(W,W')\in\mathtt W^-(N,N')$, so that $ (N\ominus W)^-
\simeq (N'\ominus W')^-. $ Shortcut removal does not change clusters. Hence the same
cluster-preservation argument as above shows that every vertex in $D_N$ must belong to $W$ and every
vertex in $D_{N'}$ must belong to $W'$. Thus $ D_N\subseteq W$ and $D_{N'}\subseteq W'. $ By
order-independence of $\ominus$, $ (N\ominus W)^- = (\widehat N\ominus(W\setminus D_N))^- $ and $
(N'\ominus W')^- = (\widehat N'\ominus(W'\setminus D_{N'}))^-. $ Therefore $ (W\setminus
D_N,W'\setminus D_{N'}) \in \mathtt W^-(\widehat N,\widehat N'), $ which implies $ \dist^-(N,N')
\geq |D_N|+|D_{N'}| + \dist^-(\widehat N,\widehat N'). $ Conversely, let $U$ and $U'$ realize
$\dist^-(\widehat N,\widehat N')$. Then $ (\widehat N\ominus U)^- \simeq (\widehat N'\ominus U')^-.
$ Using order-independence once more gives $ (N\ominus(D_N\cup U)\bigr)^- \simeq
\bigl(N'\ominus(D_{N'}\cup U')\bigr)^-. $ Hence $ \dist^-(N,N') \leq |D_N|+|D_{N'}| +
\dist^-(\widehat N,\widehat N'). $ Thus, $ \dist^-(N,N') = |D_N|+|D_{N'}| + \dist^-(\widehat
N,\widehat N'). $ Applying Lemma~\ref{lem:dc-to-dc-similar} completes the proof.
\end{proof}

We will make frequent use of the following simple result when dealing with DC-similar networks.
\begin{lemma}\label{lem:simw} Let $N$ and $N'$ be DC-similar networks with canonically
       identified vertex sets and let $W\subseteq V^0(N)$ and $W'\subseteq V^0(N')$. 
       If  $ N\ominus W \simeq N'\ominus W'$, then $W=W'$. 
       Moreover, if  $
       (N\ominus W)^- \simeq (N'\ominus W')^-$, then $W=W'$. 
\end{lemma} 
\begin{proof} 
Let $N,N'$ and $W,W'$ be as stated.
By Lemma~\ref{lem:network-ominus}, the $\ominus$-operation preserves the clusters of all
vertices that remain. Hence, $ \mathfrak C_{N\ominus W} = \{\CC_N(v)\mid v\in V(N)\setminus W\}$ and
$ \mathfrak C_{N'\ominus W'} = \{\CC_{N'}(v)\mid v\in V(N')\setminus W'\} $. 
Suppose first that $N\ominus W\simeq N'\ominus W'$. Then these two
networks have the same clustering system.  Using the canonically identified
vertex sets and the fact that $\CC_N(v)=\CC_{N'}(v)$ for every vertex $v$, we obtain $
\{\CC_N(v)\mid v\in V(N)\setminus W\} = \{\CC_N(v)\mid v\in V(N)\setminus W'\}. $ Since $N$ is
distinct-cluster, the map $v\mapsto \CC_N(v)$ is injective. Therefore $V(N)\setminus W=V(N)\setminus
W'$, and hence $W=W'$.

Now suppose that $ (N\ominus W)^- \simeq (N'\ominus W')^-. $ Since shortcut removal preserves
clusters, we have $ \mathfrak C_{(N\ominus W)^-} = \mathfrak C_{N\ominus W} $ and $ \mathfrak
C_{(N'\ominus W')^-} = \mathfrak C_{N'\ominus W'}. $ Thus the same argument as above applies and
again yields $W=W'$.
\end{proof}

Let $N$ be a distinct-cluster network on $X$. Since every vertex of $N$ is uniquely determined by
its cluster, we may identify each vertex $v\in V(N)$ with $\CC_N(v)$. Under this identification, $
V(N)=\mathfrak C_N$, and we say that $N$ is \emph{cluster-canonically represented}.

\begin{corollary}\label{cor:dist-DC-equality}
The distances $\dist$ and $\dist^-$ are genuine metrics on the class of cluster-canonically
represented distinct-cluster networks on $X$. More precisely, the isomorphism relations $\simeq$
and $\simeq^-$ in Theorems~\ref{thm:distmetric} and~\ref{thm:dist-minus-metric}, respectively, can
be replaced by equality.
In particular, this holds for every class of pairwise DC-similar networks on $X$ whose vertex sets
are canonically identified.
\end{corollary}

\begin{proof}
Let $N$ and $N'$ be cluster-canonically represented distinct-cluster networks on $X$, and let
$W\subseteq V^0(N)$ and $W'\subseteq V^0(N')$. The $\ominus$-operation preserves the clusters of all
remaining vertices. Moreover, every isomorphism between networks on $X$ fixes the leaves and
therefore preserves the cluster of every vertex. Since every remaining vertex is identified with its
cluster, any isomorphism between $N\ominus W$ and $N'\ominus W'$ fixes every remaining vertex.
Consequently, $ N\ominus W\simeq N'\ominus W'$ if and only if $N\ominus W=N'\ominus W'$. Shortcut
removal also preserves the clusters of all vertices. By the same argument, $ (N\ominus W)^-\simeq
(N'\ominus W')^-$ if and only if $(N\ominus W)^-=(N'\ominus W')^-$. Hence, the isomorphism
conditions in the definitions of $\dist$ and $\dist^-$ can be replaced by equality. The assertion
now follows from Theorems~\ref{thm:distmetric} and~\ref{thm:dist-minus-metric}.

If the networks are pairwise DC-similar, then their cluster sets, and thus their canonically
represented vertex sets, coincide. This is just a special case of the preceding statement.
\end{proof}


\section{The connection between $\dist$ and $\dist^-$ and hardwired cluster distances}
\label{sec:ominus-hardwired}

Recall that, for a network $N$, $\mathfrak C_N$ denotes the set of
clusters of $N$, whereas $\mathfrak M_N$ denotes the corresponding
multiset of clusters. Thus, a cluster $C$ occurs in $\mathfrak M_N$ with
multiplicity  $|\{v\in V(N)\mid \CC_N(v)=C\}|$.

We now relate $\dist$ and $\dist^-$ to two  clusters distances of networks $N$ and $N'$. 
First,  the \emph{Robinson-Foulds (RF) distance} $d_{RF} \coloneqq|\mathfrak C_N \Delta \mathfrak C_{N'}|$ 
has been initially introduced to compare trees \cite{robinson1981comparison}.  
Its extension to networks~\cite{cardona2008metrics} is often called the \emph{hardwired clusters distance}~\cite[Chapter 6]{huson2010phylogenetic}.
Second, we consider  the symmetric difference of cluster multiset $|\mathfrak M_N \Delta \mathfrak M_{N'}|$, which we may call the \emph{multiset hardwired clusters distance}.  
It is well-known that neither $|\mathfrak C_N \Delta \mathfrak C_{N'}|$ nor $|\mathfrak M_N \Delta \mathfrak M_{N'}|$ yields a metric, as they can both be zero on non-isomorphic networks~\cite[Figure 11]{cardona2008tripartitions}.

\begin{lemma}\label{lm:ineqs}
For all networks $N$ and $N'$ on $X$, it holds that  \[d_{RF} =|\mathfrak C_N \Delta \mathfrak C_{N'}| \leq |\mathfrak M_N \Delta \mathfrak M_{N'}|  \leq \dist^-(N,N') \leq \dist(N,N').\]
\end{lemma}
\begin{proof}
Let $N$ and $N'$ be two DAGs on $X$. In what follows, we use sets $W$ and $W'$
and always assume that $W\subseteq V^0(N)$ and $W\subseteq V^0(N')$.

To verify that $\dist^-(N,N') \leq \dist(N,N') $, let $W$ and $W'$
be such that $N \ominus W \simeq N'
\ominus W'$ and $\dist(N,N')=|W|+|W'|$. Then, we have $(N \ominus W)^- \simeq (N' \ominus W')^-$, so
$\dist^-(N,N') \leq |W|+|W'|=\dist(N,N')$.

We now show that $|\mathfrak M_N \Delta \mathfrak M_{N'}| \leq \dist^-(N,N')$. Let $W$ and $W'$ be such that $(N \ominus W)^- \simeq (N'
\ominus W')^-$ and $\dist^-(N,N')=|W|+|W'|$. Let $\varphi\colon V(N)\setminus W \to
V(N')\setminus W'$ be an isomorphism between $(N\ominus W)^-$ and $(N'\ominus W')^-$. By
Lemma~\ref{lem:properties-SF-G_NEW} and Lemma~\ref{lem:network-ominus},
 both $(N \ominus W)^-$ and $(N' \ominus W')^-$ are networks on $X$. In particular,
$\varphi(x)=x$ holds for all $x\in X$. Moreover, since by Lemma~\ref{lem:properties-SF-G_NEW} and Lemma~\ref{lem:network-ominus}, the
$\ominus$-operation preserves clusters of the vertices that remain and shortcut removal does not
change such clusters, we obtain, for all $v\in
V(N)\setminus W$, \[ \CC_N(v) = \CC_{N\ominus W}(v) = \CC_{(N\ominus W)^-}(v) = \CC_{(N'\ominus
W')^-}(\varphi(v)) = \CC_{N'\ominus W'}(\varphi(v)) = \CC_{N'}(\varphi(v)).
\]
Hence the vertices that remain after deleting $W$ from $N$ and $W'$ from
$N'$ contribute the same multiset of clusters.
Therefore, every occurrence of a cluster in
$\mathfrak M_N\setminus \mathfrak M_{N'}$ must be contributed by a vertex
$v\in W$. Similarly, every occurrence of a cluster in
$\mathfrak M_{N'}\setminus \mathfrak M_N$ must be contributed by a vertex
$v'\in W'$. 
Therefore, $|W| \geq |\mathfrak M_N \setminus \mathfrak M_{N'}|$ and $|W'| \geq |\mathfrak M_{N'}
\setminus \mathfrak M_N|$ holds. It follows that $|W|+|W'| \geq |\mathfrak M_N \setminus
\mathfrak M_{N'}|+|\mathfrak M_{N'} \setminus \mathfrak M_N|=|\mathfrak M_N \Delta \mathfrak
M_{N'}|$. This together with $\dist^-(N,N')=|W|+|W'|$ implies that 
$|\mathfrak M_N \Delta \mathfrak M_{N'}| \leq \dist^-(N,N')$.

We finally show that $|\mathfrak C_N \Delta \mathfrak C_{N'}| \leq |\mathfrak M_N \Delta
\mathfrak M_{N'}|$. 
By definition, a cluster belongs to $\mathfrak C_N$ if and only if it
occurs with positive multiplicity in $\mathfrak M_N$, and analogously for
$N'$. Hence every cluster in
$\mathfrak C_N \Delta \mathfrak C_{N'}$ occurs with positive multiplicity
in exactly one of $\mathfrak M_N$ and $\mathfrak M_{N'}$. Therefore, each
cluster in $\mathfrak C_N \Delta \mathfrak C_{N'}$ contributes at least one occurrence to the  symmetric
difference $\mathfrak M_N \Delta \mathfrak M_{N'}$. 
Consequently, $|\mathfrak C_N \Delta \mathfrak C_{N'}| \leq |\mathfrak M_N
\Delta \mathfrak M_{N'}|$ follows.
\end{proof}

Using Lemma~\ref{lm:ineqs}, there is a simple way to determine the largest possible value of
$\dist$ and $\dist^-$ when the numbers of removable vertices are fixed. This is useful, for
example, if one wants to normalize these distances.

\begin{lemma}\label{lem:diameter}
Let $X$ be a finite set with $|X|\geq 2$. For $m\geq 0$, let
$\mathcal N_X(m)$ be the class of networks on $X$ for which
$|V^0(N)|=m$.
Then, for all $m_1,m_2\geq 0$,
\[
\mathrm{diam}_X(m_1,m_2) \coloneqq \max_{\substack{N\in\mathcal N_X(m_1)\ N'\in\mathcal N_X(m_2)}} \dist(N,N')
= m_1+m_2\] and
\[\mathrm{diam}^-_X(m_1,m_2)  \coloneqq  \max_{\substack{N\in\mathcal N_X(m_1)\ N'\in\mathcal N_X(m_2)}} \dist^-(N,N') = m_1+m_2.\]
\end{lemma}
\begin{proof}
Let $N\in\mathcal N_X(m_1)$ and $N'\in\mathcal N_X(m_2)$. By Lemma~\ref{lm:stt},
$N\ominus V^0(N)$ and $N'\ominus V^0(N')$ are both star trees on $X$. 
This and Lemma~\ref{lm:ineqs} implies $
\dist^-(N,N') \leq \dist(N,N')\leq |V^0(N)|+|V^0(N')|=m_1+m_2$. 

It remains to show that this bound is sharp. Let $N$ be obtained from the star tree $T$ on $X$ with
root $\rho_T$ by inserting a directed $r\rho_T$-path of with $m_1+1$ arcs. Hence, $N$ is a network
with root $r$ and all vertices in $V^0(N)$ have cluster $X$. Let $N'$ be obtained from the star tree
on $X$ by subdividing the arc to the leaf $x$ with $m_2$ vertices. Hence, all vertices in $V^0(N')$
have cluster $\{x\}$. Since $|X|\geq 2$, we have $\{x\}\neq X$. Thus the occurrences in $\mathfrak
M_N$ contributed by the vertices in $V^0(N)$ are $m_1$ additional occurrences of $X$, while the
occurrences in $\mathfrak M_{N'}$ contributed by the vertices in $V^0(N')$ are $m_2$ additional
occurrences of $\{x\}$. The root occurrences of $X$ and the leaf occurrences of the singleton
clusters occur in both networks. This together with Lemma~\ref{lm:ineqs} implies $m_1+m_2 =
|\mathfrak M_N\Delta \mathfrak M_{N'}| \leq \dist^-(N,N') \leq \dist(N,N') \leq m_1+m_2. $
Therefore $\dist^-(N,N')=\dist(N,N')=m_1+m_2$, and the bound is sharp.
\end{proof}

Lemma~\ref{lm:ineqs} shows that the multiset hardwired cluster distance is always a lower bound for
$\dist^-$. We now identify a general situation in which this lower bound is tight. The essential
requirements are that the class ``behaves well'' under the shortcut-free $\ominus$-reductions used in the
definition of $\dist^-$, and for which the multiset of clusters determines the network up to
shortcut removal.

\begin{definition} 
Let $\Gamma$ be a class of networks on $X$. We say that $\Gamma$ is 
\emph{$\ominus$-shortcut-closed}
 if, for every $N\in\Gamma$ and every $W\subseteq V^0(N)$, 
 it holds that $ (N\ominus W)^-\in\Gamma$. 
 
Moreover, we say that $\Gamma$ is \emph{multicluster-shortcut-encoded}
 if, for all $N,N'\in\Gamma$, the following implication holds:
$ \mathfrak M_N=\mathfrak M_{N'} \implies N^-\simeq (N')^-$.
\end{definition}

The following theorem shows that the latter two properties 
are sufficient to make the lower bound $|\mathfrak M_N\Delta \mathfrak M_{N'}|$
of $\dist^-(N,N')$  from Lemma~\ref{lm:ineqs} tight.

\begin{proposition}\label{prop:symdif} 
Let $\Gamma$ be a class of networks on $X$ that is $\ominus$-shortcut-closed and
multicluster-shortcut-encoded. Then, for all $N,N'\in \Gamma$ it holds that
\[ \dist^-(N,N')=|\mathfrak M_N\Delta \mathfrak M_{N'}|.\] 
Moreover, the distance $\dist^-(N,N')$ and sets
$W\subseteq V^0(N)$ and $W'\subseteq V^0(N')$ with
$(N\ominus W)^-\simeq (N'\ominus W')^-$
and $|W|+|W'|=\dist^-(N,N')$ 
can be computed in polynomial time.
\end{proposition}
\begin{proof}
Let $\Gamma$ be as stated.
By Lemma~\ref{lm:ineqs}, $|\mathfrak M_N \Delta \mathfrak M_{N'}| \leq \dist^-(N,N')$ holds for all
DAGs $N,N' \in \Gamma$. Hence, it remains to show that $\dist^-(N,N') \leq |\mathfrak M_N \Delta \mathfrak M_{N'}|$
holds for all $N,N'\in \Gamma$.

Let $N,N'\in \Gamma$. Note that the clusters $X$ and $\{x\}$, for all $x\in X$, occur at least once
in $\mathfrak M_N$ and in $\mathfrak M_{N'}$. They may, however, occur with multiplicity greater
than one. The root of $N$ contributes one occurrence of the cluster $X$ to $\mathfrak M_N$, while
each leaf $x\in X$ contributes one occurrence of the cluster $\{x\}$ to $\mathfrak M_N$; the same
holds for $N'$. For all cluster $C\in \mathfrak M_N\cap\mathfrak M_{N'}$, let $k_C$ denote the
multiplicity of $C$ in $\mathfrak M_N\cap\mathfrak M_{N'}$. We choose $k_C$
distinct vertices $v$ of $N$ with $\CC_N(v)=C$, and $k_C$ distinct vertices $v'$ of $N'$ with
$\CC_{N'}(v')=C$. In view of the preceding observation, this can be
done in such a way that the root and all leaves of $N$, (\emph{resp.} of $N'$) are
chosen. Let $W$ be the set of vertices of $N$ that were
not chosen and define $W'$ analogously for $N'$.  Since the roots and all leaves of $N$ and $N'$ were chosen, we have
$W\subseteq V^0(N)$ and $W'\subseteq V^0(N')$. Moreover, by construction, the chosen vertices in
$N$ and in $N'$ represent precisely the cluster occurrences in $\mathfrak M_N\cap\mathfrak M_{N'}$.
Hence the vertices in $W$ and $W'$ represent precisely the cluster occurrences in $\mathfrak
M_N\setminus\mathfrak M_{N'}$ and $\mathfrak M_{N'}\setminus\mathfrak M_N$, respectively. Hence, 
$|V(N) \setminus W|=|V(N') \setminus W'|=|\mathfrak M_N\cap\mathfrak M_{N'}|$ and,
therefore $|W|+|W'|=|\mathfrak M_N \Delta \mathfrak M_{N'}| $.

 Now, consider the networks $(N
\ominus W)^-$ and $(N' \ominus W')^-$. Since $\Gamma$ is $\ominus$-shortcut-closed and since $N,N' \in \Gamma$, we
can conclude that $(N \ominus W)^- \in \Gamma$ and $(N' \ominus W')^- \in \Gamma$.
By construction of $W$, the multiset $\{\CC_N(v)\mid v\in V(N)\setminus W\}$ is precisely
$\mathfrak M_N\cap\mathfrak M_{N'}$. Since $V(N\ominus W)=V(N)\setminus W$ and $\ominus$ preserves
the clusters of all remaining vertices (cf.\ Lemma~\ref{lem:network-ominus}), it follows that
$\mathfrak M_{N\ominus W}=\mathfrak M_N\cap\mathfrak M_{N'}$. By the same argument, 
$\mathfrak M_{N'\ominus W'}=\mathfrak M_N\cap\mathfrak M_{N'}$.

Since $\ominus$ and shortcut
removal preserve the clusters of all remaining vertices (cf.\ Lemma~\ref{lem:properties-SF-G_NEW}
and Lemma~\ref{lem:network-ominus}), both $(N\ominus W)^-$ and $(N'\ominus W')^-$ have
cluster multisets $\mathfrak M_N\cap\mathfrak M_{N'}$, i.e., $\mathfrak M_{(N\ominus W)^-} =
\mathfrak M_N\cap \mathfrak M_{N'} = \mathfrak M_{(N'\ominus W')^-}$.  Since
$\Gamma$ is 
multicluster-shortcut-encoded
and since $(N \ominus W)^-$ and
$(N' \ominus W')^-$ belong to $\Gamma$, it follows that $\bigl((N\ominus W)^-\bigr)^- \simeq
\bigl((N'\ominus W')^-\bigr)^-$. Since both networks $(N\ominus W)^-$ and $(N'\ominus W')^-$ are
already shortcut-free we can conclude that $(N \ominus W)^-\simeq (N' \ominus W')^-$. Hence,
$\dist^-(N,N') \leq |W|+|W'|$. Since $|W|+|W'|=|\mathfrak M_N \Delta \mathfrak M_{N'}|$, we can
conclude that $\dist^-(N,N') \leq |\mathfrak M_N \Delta \mathfrak M_{N'}|$.

We now show that $\dist^-(N,N')$ can be computed in polynomial time in $|V(N)|+|V(N')|$. To
this end, it suffices to show that $|\mathfrak M_N\Delta\mathfrak M_{N'}|$ can be computed in
polynomial time in the size of the input networks. Indeed, one can compute all clusters of $N$ in a
single bottom-up traversal of the DAG: starting at the leaves, assign the cluster $\{x\}$ to each
leaf $x\in X$, and then, in post order, assign to every inner vertex $v$ the union of the clusters
of its children. 
This constructs $\mathfrak M_N$ in polynomial time. The same procedure constructs $\mathfrak M_{N'}$. 
Finally, the multiset symmetric
difference of $\mathfrak M_N$ and $\mathfrak M_{N'}$ can be computed by sorting the obtained cluster
representations. Hence $|\mathfrak M_N\Delta\mathfrak M_{N'}|$, and therefore $\dist^-(N,N')$
can be computed in polynomial time. 
Now, for every cluster $C$, we can compare its multiplicities in $\mathfrak M_N$ and $\mathfrak
M_{N'}$. If $m_N(C)>m_{N'}(C)$, we choose $m_N(C)-m_{N'}(C)$ vertices $v\in V(N)$ with $\CC_N(v)=C$
and put them into $W$; if $m_{N'}(C)>m_N(C)$, we choose $m_{N'}(C)-m_N(C)$ vertices $v'\in V(N')$
with $\CC_{N'}(v')=C$ and put them into $W'$. As above, the occurrences of $X$ contributed by the
roots and the occurrences of ${x}$ contributed by the leaves are kept in both networks. Therefore
the chosen sets satisfy $W\subseteq V^0(N)$ and $W'\subseteq V^0(N')$. Thus the construction of $W$
and $W'$ described above can be carried out in polynomial time. Since $|W|+|W'|=|\mathfrak
M_N\Delta\mathfrak M_{N'}|=\dist^-(N,N')$ and, as argued above, $(N \ominus W)^-\simeq (N' \ominus W')^-$, 
optimal deletion sets $W,W'$ can be computed in polynomial time.
\end{proof}

We next record several classes to which Proposition~\ref{prop:symdif} applies, see Figure~\ref{fig:trees-regular} for
an example.
The first class, namely networks satisfying (PCC), will be particularly
useful, since every subclass of (PCC) inherits the equality 
$\dist^-(N,N')=|\mathfrak M_N\Delta \mathfrak M_{N'}|$
obtained in Proposition~\ref{prop:symdif}.

\begin{proposition}\label{prop:Gamma-examples}
The following classes of networks are $\ominus$-shortcut-closed and multicluster-shortcut-encoded. 
\begin{enumerate}[label=\emph{(\roman*)},noitemsep]
    \item The class of networks  on $X$ satisfying (PCC).
    \item The class of semi-regular networks  on $X$.
    \item The class of regular networks  on $X$. 
\end{enumerate}
\end{proposition}
\begin{proof}
We first consider the class $\Gamma$ of networks on $X$ satisfying (PCC). Let $N\in \Gamma$, and let $W\subseteq V^0(N)$. By
Lemma~\ref{lem:properties-SF-G_NEW}  and Lemma~\ref{lem:network-ominus}, the $\ominus$-operation
as well as shortcut-removal preserves both clusters and the
ancestor relation among the vertices that remain. Hence, 
if $C_{(N \ominus W)^-}(u) \subseteq C_{(N \ominus W)^-}(v)$, then $C_N(u) \subseteq C_N(v)$, which together with the fact that $N$ satisfies (PCC) 
implies that $u$ and $v$ are $\preceq_N$-comparable, and therefore $\preceq_{(N \ominus W)^-}$-comparable as well.  We deduce that  $(N\ominus W)^-$
satisfies (PCC) and it follows that $(N\ominus W)^-\in \Gamma$.
Hence, $\Gamma$ is $\ominus$-shortcut-closed. 
It remains to show that $\Gamma$ is multicluster-shortcut-encoded. Let $N,N'\in  \Gamma$ and
assume that $\mathfrak M_N=\mathfrak M_{N'}$. 
Since by Lemma~\ref{lem:properties-SF-G_NEW}, shotcut-removal preserves ancestors
relationsship among the vertices and thus preserves the clusters, we can conclude that 
$\mathfrak M_{N^-}=\mathfrak M_N=\mathfrak M_{N'}=\mathfrak M_{(N')^-}$.  
Moreover, $N^-$ and $(N')^-$ are shortcut-free and satisfy (PCC), hence
they are semi-regular. By \cite[Thm~5]{Hellmuth2023}, a semi-regular
network is uniquely determined, up to isomorphism, by its cluster
multiset. Consequently, $N^-\simeq (N')^-$ 
and we can conclude that $\Gamma$ is  multicluster-shortcut-encoded.

The assertion for semi-regular networks follows immediately from the preceding paragraph. Indeed,
semi-regular networks are precisely the shortcut-free networks satisfying (PCC). If $N$ is
semi-regular and $W\subseteq V^0(N)$, then the preceding paragraph shows that $(N\ominus W)^-$
satisfies (PCC); by construction, it is shortcut-free. Thus $(N\ominus W)^-$ is semi-regular. Hence
the class of semi-regular networks is $\ominus$-shortcut-closed. By the same arguments 
used in the previous paragraph, the class of semi-regular networks
is  multicluster-shortcut-encoded.

Finally, consider the class $\Gamma$ of regular networks on $X$. By \cite[Thm.~2]{Hellmuth2023}, a network is
regular if and only if it is semi-regular and has no vertex of outdegree one. Let $N\in \Gamma$ and
let $W\subseteq V^0(N)$. Since regular networks are semi-regular, the semi-regular case implies that
$(N\ominus W)^-$ is semi-regular. It remains to show that $(N\ominus W)^-$ has no vertex of
outdegree one. Suppose, for contradiction, that $(N\ominus W)^-$ has a vertex $u$ with outdegree
one, and let $v$ be its unique child. Then $ \CC_{(N\ominus W)^-}(u) = \CC_{(N\ominus W)^-}(v)$.
Since $u$ and $v$ are vertices of $(N\ominus W)^-$, they are also vertices of $N$. Moreover,
$\ominus$ and shortcut removal preserve, by Lemma~\ref{lem:properties-SF-G_NEW} and
Lemma~\ref{lem:network-ominus}, clusters of the vertices that remain. Hence $ \CC_N(u) =
\CC_{(N\ominus W)^-}(u) = \CC_{(N\ominus W)^-}(v) = \CC_N(v)$. But $u\neq v$, since $(u,v)$ is an
arc of $(N\ominus W)^-$. Thus, there are distinct vertices of $N$ having the same clusters; a
contradiction to $N$ being regular. Therefore $(N\ominus W)^-$ has no vertex of outdegree one. By
\cite[Thm.~2]{Hellmuth2023}, it is regular. Thus $\Gamma$ is $\ominus$-shortcut-closed. It remains
to show that $\Gamma$ is multicluster-shortcut-encoded. Let $N,N'$ be regular and assume that
$\mathfrak M_N=\mathfrak M_{N'}$. In a regular network, distinct vertices have distinct clusters.
Hence $\mathfrak M_N=\mathfrak C_N$ and $\mathfrak M_{N'}=\mathfrak C_{N'}$, when cluster sets are
regarded as multisets. Therefore $\mathfrak C_N=\mathfrak C_{N'}$. Since regular networks are
uniquely determined by their clustering systems, it follows that $N\simeq N'$. Regular networks are
shortcut-free, and hence $N^- =N \simeq N' =(N')^-$. Thus, $\Gamma$ is multicluster-shortcut-encoded. 	
\end{proof}

\begin{theorem}\label{thm:dist-minus-easy}
Let $\Gamma$ be either of the following classes of networks with the same leaf set $X$:
\begin{enumerate}[label=\emph{(\roman*)},noitemsep]
	\item The class of networks satisfying (PCC)
		\item The class of rooted phylogenetic trees.
		\item The class of phylogenetic level-1 networks
		\item The class of binary level-1 networks. 
		\item The class of tree-child networks.
		\item The class of normal networks.
		\item The class of semi-regular networks;
    \item The class of regular networks.
\end{enumerate}
Then, for all $N,N'\in \Gamma$ it holds that 
\[\dist^-(N,N')  = |\mathfrak M_N \Delta \mathfrak M_{N'}|.\] 
In particular, for all $N,N'\in \Gamma$, 
$\dist^-(N,N')$ and sets $W\subseteq V^0(N)$ and $W'\subseteq V^0(N')$ with
$(N\ominus W)^-\simeq (N'\ominus W')^-$
and $|W|+|W'|=\dist^-(N,N')$ 
can be computed in polynomial time in the size of the networks $N$ and $N'$.

\end{theorem}
\begin{proof}
By Proposition~\ref{prop:Gamma-examples}, the class of networks satisfying (PCC), the class of
semi-regular networks and of the class regular networks are all $\ominus$-shortcut-closed and
multicluster-shortcut-encoded. This together with Proposition~\ref{prop:symdif} implies that, for all
networks $N$ and $N'$ in these classes it holds that $\dist^-(N,N') = |\mathfrak M_N \Delta
\mathfrak M_{N'}|$. 

For the remaining classes, the assertion follows because they are subclasses of the class of
networks satisfying (PCC) \cite{Hellmuth2023}. 
Thus, if $N,N'\in\Gamma$ for any of the classes \emph{(ii)--(vi)}, then
$N$ and $N'$ satisfy (PCC), and the equality follows from the case \emph{(i)}.

The last statement, concerning polynomial-time constructions follows from Proposition~\ref{prop:symdif}.
\end{proof}

\begin{figure}
	\centering
	\includegraphics[width = 0.9\textwidth]{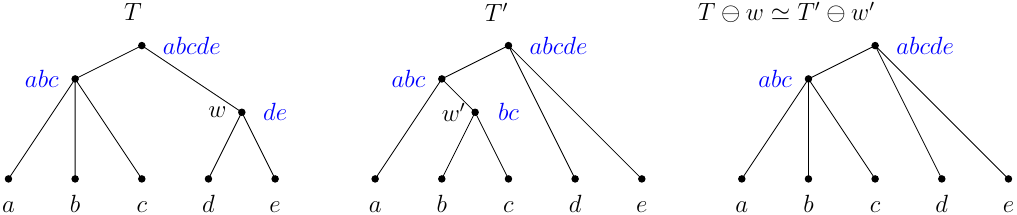}
	\caption{A phylogenetic tree $T$ and $T'$ together with the phylogenetic
	tree $T\ominus w \simeq T\ominus w'$. 
	 Both $T$ and $T'$ are regular. 
	The cluster of non-leaf vertices are indicated next to the respective vertex. 
	Following the construction used in Proposition~\ref{prop:symdif}, 
	the vertices corresponding to the clusters in
	$\mathfrak M_N \Delta \mathfrak M_{N'} = \{\{bc\},\{de\}\}$ are 
	$w\in V^0(T)$ and 	$w'\in V^0(T)$. In particular, 
	$\dist^-(T,T') = |\{w\} + |\{w'\}|=2$ which is in accordance
	with Theorem~\ref{thm:dist-minus-easy}, that states that 
	$\dist^-(T,T')  = |\mathfrak M_N \Delta \mathfrak M_{N'}|$. 
	Corollary~\ref{cor:dist-RF} implies 
	$\dist^-(T,T')  = |\mathfrak C_N \Delta \mathfrak C_{N'}|$. 
    By Corollary~\ref{cor:dist-RF-tree} even	$\dist^-(T,T')=	\dist(T,T')$ holds. }
	
\label{fig:trees-regular}
\end{figure}

As a direct consequence of Theorem~\ref{thm:dist-minus-easy} and since 
two networks $N$ and $N'$ satisfy $N^-\simeq (N')^-$ precisely if $\dist^-(N,N')=0$, 
we obtain
\begin{theorem}\label{thm:sf-iso-easy}
	For every class $\Gamma$ listed in Theorem~\ref{thm:dist-minus-easy} and all $N,N'\in\Gamma$, 
	it can be decided in polynomial time in the size of $N$ and $N'$ whether 
	their shortcut-free versions are isomorphic, that is, if  $N^-\simeq (N')^-$ . 
\end{theorem}

A further simple consequence of Theorem~\ref{thm:dist-minus-easy} is as follows. 
\begin{corollary}\label{cor:dist-RF}
Let $\Gamma$ be either of the following classes:
\begin{itemize}[noitemsep]
	\item The class of regular networks.
	\item The class of phylogenetic shortcut-free level-1 networks without retriculation vertices having out-degree one.
	\item The class of phylogenetic normal networks without retriculation vertices having out-degree one.
	\item The class of phylogenetic trees.
\end{itemize}
Then, $\dist^-(N,N')  = d_{RF}(N,N')$ for all $N,N'\in \Gamma$. 
\end{corollary}
\begin{proof}
We first consider the class of regular networks. Since regular networks are distinct-cluster
networks, every cluster of a regular network $N$ appears with multiplicity one in
$\mathfrak M_N$. Hence, for regular networks $N$ and $N'$, we have
$d_{RF}(N,N')=|\mathfrak C_N\Delta\mathfrak C_{N'}|=|\mathfrak M_N\Delta\mathfrak M_{N'}|$
This and Theorem~\ref{thm:dist-minus-easy} implies that 
$d_{RF}(N,N') = \dist^-(N,N')$.

It remains to observe that the other listed classes are subclasses of the class of regular
networks. Indeed, phylogenetic trees are regular by \cite[Cor.~9]{Hellmuth2023}, and
phylogenetic shortcut-free level-1 networks without reticulation vertices of out-degree one are
regular by \cite[Prop.~15]{Hellmuth2023}. Moreover,
\cite[Thm.~4.11]{HLM:26-RegNormArxiv} implies that phylogenetic normal networks without
reticulation vertices of out-degree one, called strongly normal networks in
\cite{HLM:26-RegNormArxiv}, are regular. Thus the assertion follows from the regular-network case.
\end{proof}

\begin{corollary}\label{cor:dist-RF-tree}
Let $\Gamma$ be the class 
of phylogenetic trees.
Then,
$
\dist(T,T') = \dist^-(T,T') = d_{RF}(T,T')
$
for all $T,T'\in \Gamma$. 
\end{corollary}
\begin{proof}
By Corollary~\ref{cor:dist-RF}, we have $\dist^-(T,T') = d_{RF}(T,T')$. It remains to show that, for
phylogenetic trees, $\dist$ and $\dist^-$ coincide. Let  $T$ be a  phylogenetic tree
on $X$. For every $W\subseteq V^0(T)$, the graph $T\ominus W$ is again a phylogenetic tree on $X$, 
since deleting a vertex $v\in V^0(T)$ and reconnecting its unique parent with all children of $v$
preserves acyclicity, preserves the unique root, and does not create any vertex of indegree greater
than one. Thus no reticulation vertex can arise. Moreover, the leaf set remains $X$, since no leaf
is removed. Hence $T\ominus W$ is a rooted tree on $X$.
In particular, $T\ominus W$ contains no shortcuts since otherwise $T\ominus W$ would contain 
reticulation vertices. Hence, 
$(T\ominus W)^- = T\ominus W$ for all phylogenetic trees $T$. Hence, for phylogenetic trees
$T$ and $T'$ it holds that  $\dist(T,T') = \dist^-(T,T')$.
\end{proof}

\section{Hardness results for computing $\dist$}
\label{sec:dist-hard}

We now turn to further computational aspects. First observe that, for general
networks $N$ and $N'$, deciding whether $\dist(N,N')=0$ is precisely the
problem of deciding whether $N\simeq N'$. 
Testing whether two networks are graph isomorphic is a difficult
problem in general; in particular, Cardona et al.~\cite{cardona2014comparison}
showed that deciding isomorphism for tree-sibling time-consistent networks is
as hard as the general graph isomorphism problem.
The analogous issue also arises for $\dist^-$. 
Indeed, deciding whether
$\dist^-(N,N')=0$ is equivalent to deciding whether $N^-\simeq (N')^-$. Hence,
even for shortcut-free networks, where $N=N^-$ and $N'=(N')^-$, this contains
the network-isomorphism problem as a special case.
We are not aware of a
reference that states graph-isomorphism hardness explicitly for the restricted
class of shortcut-free networks, but the above observation shows that computing
$\dist^-$ cannot be expected to avoid isomorphism-type difficulties in general.

Nevertheless, the previous section shows that $\dist^-$ can be computed
in polynomial time for a broad range of classical network classes.
This naturally raises the question of how difficult it is to compute the
``unrelaxed'' distance $\dist$. We show below that computing $\dist$ is
NP-complete, even for a highly restricted class of tree-child networks for
which $\dist^-$ is polynomial-time computable.

To recall, 
a network $N$ is distinct-cluster if, for any $u, v \in V(N)$ it holds that $C_N(u) =
C_N(v)$ if and only if $u = v$ and two  networks $N, N'$ are DC-similar 
if they are both distinct-cluster and $\mathfrak C(N) = \mathfrak C(N')$.

For the following result we will employ the NP-complete \textsc{Set Cover} problem \cite{GareyJohnson1979}.
To recall, \textsc{Set Cover} takes as input a collection of sets $\S = \{S_1, \ldots, S_m\}$ whose elements are from a universe $U = \{u_1, \ldots, u_n\}$, 
along with an integer $k$. 
The question is whether there exists a subcollection $\S^*\subseteq \S$ with $|\S^*|\leq k$ such that every
element $u\in U$ is contained in at least one set of $\S^*$. 
The problem \textsc{Set Cover} is W[2]-hard when parameterized by the
solution size $k$ \cite{DowneyFellows2013}. Hence, unless
FPT = W[2], there is no algorithm for \textsc{Set Cover}
with running time $f(k)\cdot |I|^{O(1)}$ for any computable function $f$,
where $|I|$ denotes the input size. Moreover, W[2]-hardness is preserved
under parameterized reductions: if a problem $A$ parameterized by $k$
parameterized-reduces to a problem $B$ parameterized by $k'$, and
$k'\leq g(k)$ for some computable function $g$, then $B$ is W[2]-hard.

\begin{theorem}\label{thm:ominus-hard-similar-normal}
Let $N$ and $N'$ be DC-similar tree-child networks with canonically identified vertex sets, and let
$\delta\geq 0$ be an integer. Then the problem of deciding, given $N$, $N'$, and $\delta$, whether $
d_{\ominus}(N,N')\leq \delta $ is NP-complete and W[2]-hard when parameterized by $\delta$.
\end{theorem}
\begin{proof}
Let $N$ and $N'$ be DC-similar tree-child networks with canonically identified vertex sets, and let
$\delta\geq 0$ be an integer.

Membership in NP is immediate. A certificate is a set
$W\subseteq V^0(N)=V^0(N')$ such that $|W|+|W|\leq \delta$, that is,
$2|W|\leq \delta$. Given such a set $W$, we can construct
$N\ominus W$ and $N'\ominus W$ in polynomial time. Since $N$ and $N'$ have
canonically identified vertex sets, it remains only to check whether the
identity map on $V(N)\setminus W$ is a leaf-fixing isomorphism between
$N\ominus W$ and $N'\ominus W$. This can be verified in polynomial time.

For the NP-hardness proof, we reduce from \textsc{Set Cover}. Let $(\S,U,k)$ be an instance of
\textsc{Set Cover}. We may assume without loss of generality that $U\neq\emptyset$ and that $
\cup_{S\in\S} S = U$.
Indeed, if some element of $U$ is not contained in any set of $\S$, then the
instance is a trivial \textsc{No}-instance.
We make a copy of $\S$ and put  $\S'=\{S'\mid S\in\S\}.$
Moreover, let $r$ and $z$ be two distinct objects that do not occur in
$\S\cup U\cup \S'$.

We first construct a network $N$ on the leaf set $X$ with vertex set
\[
    V(N)=\{r\}\cup \S \cup U\cup \S' \cup \{z\}
    \quad \text{and} \quad
    X=U\cup \S' \cup \{z\},
\]
where $r$ will serve as the root of $N$. 
We add arcs from $r$ to every vertex in $\S\cup \{z\}$. For each
$S\in\S$, we add the arc $(S,S')$, where $S'\in\S'$ denotes the copy of
$S$. Moreover, for each $S\in\S$ and every $u\in S$, we add the arc
$(S,u)$. Since $\bigcup_{S\in\S} S=U$, every $u\in U$ is reachable from
$r$, that is, $u\preceq_N r$. Hence it is easy to verify that $N$ is a
network on $X$ with root $r$ and that $V^0(N)=\S$.
Moreover, each vertex $S\in\S$ has the ``private'' leaf child $S'$, and no
vertex $S\in\S$ is adjacent to $z$. In particular, the leaf $z$ is only
adjacent to the root $r$. Moreover, $\CC_N(r)=X$. 
It is now straightforward to verify that $N$ is distinct-cluster and tree-child (the latter because
every non-leaf vertex of $N$ has a child that is a leaf of indegree 1). Finally, $N$ is
shortcut-free since the only arcs are of the form $(r,S)$, $(r,z)$, $(S,S')$, and $(S,u)$ with $u\in S$,
and none of these arcs admits an alternative directed path with the same endpoints.

Now let $\widehat N$ be obtained from $N$ by adding the arcs $(r,u)$ for all $u\in U$. 
The networks $N$ and $\widehat N$ are illustrated in Figure~\ref{fig:NP}.
One easily observes that,  $V(N)=V(\widehat N)$ and $V^0(\widehat N)=\S$. 
Since
$\cup_{S\in\S}S=U$, for every $u\in U$ there is some $S\in\S$ with $u\in S$
and therefore some $ru$-path $r\to S\to u$ in $\widehat N$. Thus, each added arc $(r,u)$ is a shortcut in $\widehat N$. Since $N$ is
shortcut-free and these are the only arcs added to $N$, it follows that $ \widehat N^- = N $. In
particular, by Lemma~\ref{lem:properties-SF-G_NEW}, the added shortcuts do not change any ancestor
relation and therefore do not change any cluster. Hence $N$ and $\widehat N$ have the same
clustering system. In particular, $\CC_N(v)=\CC_{\widehat N}(v)$ for all $v\in V(N)=V(\widehat N)$. 
Hence, since $N$ is distinct-cluster, $\widehat N$ is
distinct-cluster as well. Moreover, $\widehat N$ remains tree-child, because no vertex or arc of $N$
was removed and since no arc that contains $S'\in \S'$ or $z$ was added. In summary, $N$ and $\widehat N$ are 
DC-similar tree-child networks with canonically identified vertex sets.

\begin{figure}[h]
    \centering
    \includegraphics[width=0.6\textwidth]{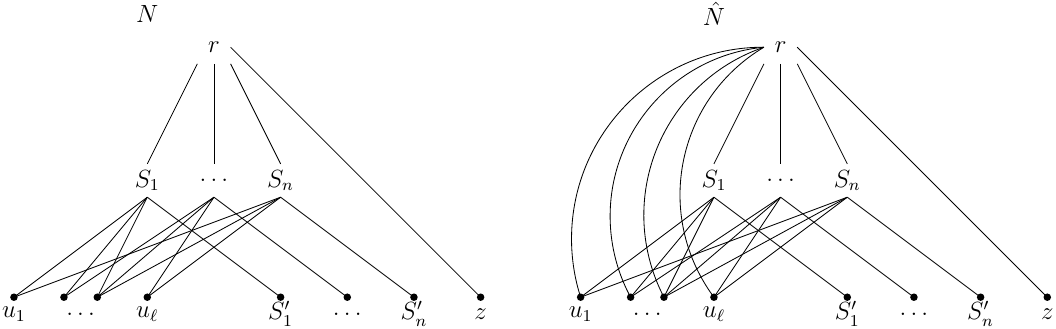}
    \caption{Shown are two networks $N$ and $\widehat N$ that serve as generic instances for which the problem
    				of deciding as whether $d_{\ominus}(N,N')\leq \delta $  for some integer $\delta \geq 0$
    				is NP-hard, see Theorem~\ref{thm:ominus-hard-similar-normal}.  }
    \label{fig:NP}
\end{figure}

We claim that the \textsc{Set Cover} instance has a solution of size at most $k$ if and only if $
d_{\ominus}(N,\widehat N)\leq 2k$.

Suppose first that $\S^*\subseteq \S$ is a set cover of $U$ with $|\S^*|\leq k$. We show that $
N\ominus \S^*=\widehat N\ominus \S^* $. By construction, $V(N)=V(\widehat N)$ holds and thus,
$V(N)\setminus \S^* = V(\widehat N)\setminus \S^*$. Thus, it remains to compare their arc sets. In
$N$, every vertex $S\in\S^*$ has the unique parent $r$ and children $S'$ and $u$ for all $u\in S$.
Hence applying $\ominus$ to the vertices  $S\in \S^*$ removes $S$ and its incident arcs
and adds the arc $(r,S')$ as well as the arcs $(r,u)$ for all $u\in S$. Since $\S^*$
covers $U$, the arc $(r,u)$ is present in $N\ominus \S^*$ for every $u\in U$. These arcs $(r,u)$,
$u\in U$, are precisely the additional arcs of $\widehat N$ compared with $N$, and they remain
present in $\widehat N\ominus \S^*$. The arcs $(r,S')$, $S\in\S^*$, are created in both $N\ominus
\S^*$ and $\widehat N\ominus \S^*$. Therefore, $ N\ominus \S^*=\widehat N\ominus \S^* $.
Consequently, $ d_{\ominus}(N,\widehat N) \leq |\S^*|+|\S^*| \leq 2k $.

Conversely, suppose that $d_{\ominus}(N,\widehat N)\leq 2k$. Since $N$ and $\widehat N$ are DC-similar
networks with canonically identified vertex sets, Lemma~\ref{lem:simw} implies that the same set of
vertices must be deleted from both networks. Hence there is a set $W\subseteq V^0(N)=\S$ such that $
|W|+|W|\leq 2k $ and $ N\ominus W\simeq \widehat N\ominus W$. In particular, $|W|\leq k$.
 
We show that $W$ is a set cover of $U$. Let $u\in U$. The arc $(r,u)$ is present in $\widehat N$ by
construction. Since $W\subseteq \S =V^0(\widehat N)$, neither $r$ nor $u$ is removed, and hence $(r,u)$ remains
present in $\widehat N\ominus W$. Since isomorphisms are leaf-fixing and preserve the unique root,
and since the vertices are canonically identified, the arc $(r,u)$ must also be present in $N\ominus
W$. However, the arc $(r,u)$ is not present in $N$. The only way it can be created by applying
$\ominus$ is by deleting some vertex $S\in W\subseteq \S$ with arcs $ r\to S $ and $ S\to u$ in $N$.
By construction, the latter is equivalent to $u\in S$. Hence there exists some $S\in W$ with $u\in
S$. Since $u\in U$ was arbitrary, $W$ covers $U$. Together with $|W|\leq k$, this gives a solution
of size at most $k$ for the \textsc{Set Cover} instance.

We have therefore shown that the \textsc{Set Cover} instance has a
solution of size at most $k$ if and only if
$
    d_{\ominus}(N,\widehat N)\leq 2k.
$
This proves NP-hardness. Since the parameter value $\delta=2k$ depends
only on $k$, and \textsc{Set Cover} is W[2]-hard when parameterized by
$k$, the same reduction proves W[2]-hardness when parameterized by
$\delta$.
\end{proof}

The hardness proof above does not carry over to $\dist^-$. Indeed,
the construction exploits precisely the shortcut arcs added to $N$ in order
to obtain $\widehat N$. Since these arcs disappear under shortcut removal,
we have $\widehat N^-=N$ and therefore $d^-_{\ominus}(N,\widehat N)=0$ 
for the corresponding instances. Thus a hardness result for
$\dist^-$ would require a different construction.


\section{Computing $\dist^-$ via Bad Ancestry Graphs and Vertex Covers}
\label{sec:distMinus-VC}

In the previous section, we showed that computing $\dist$ is a difficult computational problem and
that its NP-hardness leaves little room for efficient exact or approximation algorithms. This
motivates us to focus on $\dist^-$. Although we show in
Section~\ref{sec:NPHard-dist-minus} that computing $\dist^-$ is NP-hard as well, the algorithmic
situation for $\dist^-$ is considerably more favorable.
In particular, for DC-similar networks, $\dist^-$ admits an elegant reformulation in terms of
\emph{bypassing sets}, that is, sets of vertices whose removal resolves all disagreements between
the ancestor relations of the two networks. These disagreements are captured by the
\emph{bad ancestry graph}, and minimum bypassing sets correspond precisely to minimum vertex covers
of this graph. This connection yields a structural characterization of $\dist^-$ and allows us to
transfer well-known fixed-parameter and approximation algorithms for \textsc{Vertex Cover} directly
to the computation of $\dist^-$.

\begin{definition}
Let $N$ and $N'$ be two networks on the same leaf set $X$
Then, $(W,W')$ is a \emph{bypassing pair} for $N$ and $N'$
if $W \subseteq V^0(N)$ and  $W' \subseteq V^0(N')$ and there exists an
ancestor-preserving map $\phi\colon V(N) \setminus W \to V(N') \setminus W'$.
If $(W,W')$ is a bypassing pair for $N$ and $N'$ and $W=W'$, then $W$ is
called a \emph{bypassing set} for $N$ and $N'$.

A bypassing pair $(W,W')$ for $N$ and $N'$ is called \emph{minimum} if
$|W|+|W'|$ is minimum among all bypassing pairs for $N$ and $N'$. Similarly,
a bypassing set $W$ for $N$ and $N'$ is called \emph{minimum} if $|W|$ is
minimum among all bypassing sets for $N$ and $N'$.
\end{definition}

We emphasize that a bypassing pair $(W,W')$ for
all networks $N$ and $N'$ on the same leaf-set always exists. 
To see this, let $N$ and $N'$ be two networks on the same set of leaf set $X$
and with roots $\rho_N$ and $\rho_{N'}$, respectively.
Observe that any ancestor-preserving map $\phi$ must satisfy
$\phi(\rho_N)=\rho_{N'}$ and $\phi(x)\in X$ for all $x\in X$. 
In particular, $x\preceq_N \rho_N$ and thus, $\phi(x) \preceq_{N'} \phi(\rho_N)$
holds for all $x\in X$. 
It is now easy to see that $(V^0(N),V^0(N'))$ is a bypassing pair for $N$ and $N'$
since $V(N) \setminus V^0(N) = \{\rho_N\}\cup X$ and $V(N') \setminus V^0(N') = \{\rho_{N'}\}\cup X$.

Interestingly, bypassing pairs are exactly the vertex sets
whose removal makes the corresponding shortcut-free $\ominus$-reductions isomorphic.

\begin{proposition}\label{pr:bpset}
Let $N$ and $N'$ be two networks on the same leaf set, and let $W \subseteq V^0(N)$ and $W' \subseteq
V^0(N')$. Then, $(N \ominus W)^- \simeq (N' \ominus W')^-$ if and only if 
$(W,W')$ is a bypassing pair
for $N$ and $N'$.
\end{proposition}

\begin{proof}
Let $N,N'$ be two networks on the same leaf set, and let $W \subseteq V^0(N)$ and $W' \subseteq
V^0(N')$.
Suppose first that $(N \ominus W)^- \simeq (N' \ominus W')^-$ and thus, that there is a graph
isomorphism $\phi\colon V((N \ominus W)^-) \to V((N' \ominus W')^-)$. Note that $V((N \ominus W)^-)=V(N)
\setminus W$ and $V((N' \ominus W')^-)=V(N') \setminus W'$, so $\phi$ is a bijection between $V(N)
\setminus W$ and $V(N') \setminus W'$. Now, let $u,v \in V(N) \setminus W$. 
 It holds that 
$u \preceq_N v$ if and only if $u \preceq_{(N \ominus W)^-} v$ (cf.\ Lemma~\ref{lem:network-ominus}),
which is, if and only if $\phi(u) \preceq_{(N' \ominus W')^-} \phi(v)$
which, by Lemma~\ref{lem:network-ominus}, is equivalent to
$\phi(u)\preceq_{N'}\phi(v)$ since 
$\phi$ is a graph isomorphism between $(N \ominus W)^-$ and $(N' \ominus W')^-$.
 
 Since
this holds for any two $u,v \in V(N) \setminus W$ it follows that $\phi$ is an ancestor-preserving map. 
Hence, $(W,W')$ is a bypassing pair for $N$ and $N'$.

Conversely, suppose that $(W,W')$ is a bypassing pair for $N$ and $N'$.
By definition, there exists an ancestor-preserving map $\phi\colon V(N) \setminus W \to V(N') \setminus
W'$.  Since $V(N) \setminus W=V((N \ominus W)^-)$ and $V(N') \setminus W'=V((N' \ominus W')^-)$
it follows that  
$\phi \colon V((N \ominus W)^-) \to V((N' \ominus W')^-)$ is  an ancestor-preserving map. 
This together with the fact that $(N \ominus W)^-$ and $(N' \ominus W')^-$  
are shortcut-free implies together with  Lemma~\ref{lm:isom} 
that $\phi$ is a graph-isomorphism between $(N \ominus W)^-$ and $(N' \ominus W')^-$. Hence, we have
$(N \ominus W)^- \simeq (N' \ominus W')^-$ as desired.
\end{proof}

\begin{corollary}
Let $N$ and $N'$ be two networks on the same leaf set, and let $W \subseteq V^0(N)$ and $W' \subseteq
V^0(N')$. If $N \ominus W \simeq N' \ominus W'$, then $(W,W')$ is a bypassing pair for $N$ and
$N'$.
\end{corollary}
\begin{proof}
If $N \ominus W \simeq N' \ominus W'$, then $(N \ominus W)^- \simeq (N' \ominus W')^-$ 
also holds. 
The assertion follows now from Proposition~\ref{pr:bpset}. 
\end{proof}

A further important consequence of Proposition~\ref{pr:bpset} and
Lemma~\ref{lem:simw} is the following.

\begin{corollary}\label{cor:minbypassing-SET}
Let $N$ and $N'$ be  DC-similar networks with canonically identified
vertex sets $V$. Then a pair $(W,W')$ is a bypassing pair for $N$ and $N'$ if
and only if $W=W'$ and $W$ is a bypassing set for $N$ and $N'$.
\end{corollary}
\begin{proof}
Let $(W,W')$ be a bypassing pair for $N$ and $N'$. By
Proposition~\ref{pr:bpset}, we have
$(N\ominus W)^-\simeq (N'\ominus W')^-$.
Since $N$ and $N'$ are DC-similar with canonically identified vertex sets,
Lemma~\ref{lem:simw} implies $W=W'$. Hence $W$ is a bypassing set for
$N$ and $N'$.
Conversely, if $W$ is a bypassing set for $N$ and $N'$, then, by
definition, $(W,W)$ is a bypassing pair for $N$ and $N'$.
\end{proof}

A direct consequence of Corollary~\ref{cor:minbypassing-SET} is
\begin{corollary}\label{cor:min-bypassing-set}
Let $N$ and $N'$ be DC-similar networks with canonically identified vertex
sets, and let $(W,W')\in \mathtt W^-(N,N')$. Then $(W,W')$ attains the
minimum in the definition of $\dist^-(N,N')$, that is,
$
    \dist^-(N,N')=|W|+|W'|$
if and only if $W=W'$ and $W$ is a minimum bypassing set for $N$ and $N'$.
\end{corollary}
\begin{proof}
Since $(W,W')\in\mathtt W^-(N,N')$, Proposition~\ref{pr:bpset} implies
that $(W,W')$ is a bypassing pair for $N$ and $N'$. By
Corollary~\ref{cor:minbypassing-SET}, this is equivalent to $W=W'$ and
$W$ being a bypassing set. Moreover, minimizing $|W|+|W'|$ over
$\mathtt W^-(N,N')$ is, by Corollary~\ref{cor:minbypassing-SET}, equivalent to minimizing
$2|W|$ over all bypassing sets. Hence $(W,W')$ attains the minimum in the
definition of $\dist^-(N,N')$ if and only if $W=W'$ and $W$ is a minimum
bypassing set.
\end{proof}

By Corollary~\ref{cor:min-bypassing-set}, computing $\dist^-$ for two
DC-similar networks $N$ and $N'$ amounts to finding a minimum-size
bypassing set for $N$ and $N'$. We now give an equivalent
graph-theoretical formulation in terms of vertex covers.

Let $N$ and $N'$ be DC-similar networks with canonically identified vertex
sets. A pair of distinct vertices $(u,v)$ is called a \emph{bad ancestry
pair}  of $N$ and $N'$, if $v$ is a descendant of $u$ in exactly one of
$N$ and $N'$. In
other words, $(u, v)$ is a bad ancestry pair if $N$ and $N'$ disagree on their ancestor relationship
between $u$ and $v$.

A simple consequence of the definition of bypassing sets is that they correspond to destroying each bad ancestry pair.

\begin{lemma}\label{lem:bypass-at-least-one}
	Let $N, N'$ be DC-similar networks with canonically identified vertex sets. Then $W \subseteq
	V^0(N)$ is a bypassing set for $N$ and $N'$ if and only if $W$ contains at 
	least one vertex of every bad ancestry pair of $N$ and $N'$.
\end{lemma}
\begin{proof} 
Let $N$ and $N'$ be as stated and suppose that $W$ is a bypassing set for $N$ and $N'$. Let
$(u, v)$ be a bad ancestry pair of $N$ and $N'$. If $W$ does not contain $u$ nor $v$, then
 $u, v \in V(N) \setminus W$. 
In particular, $v \prec_N u$ but $v \not \prec_{N'} u$, or
$v \prec_{N'} u$ but $v \not \prec_{N} u$. Either way, this would contradict the fact that
$W$ is a bypassing set. Therefore $W$ must contain at least one of $u$ or $v$. Conversely, assume 
that $W$ contains at least one vertex of every bad ancestry pair. Let 
$u,v \in V(N) \setminus W$ be distinct vertices. 
Since $u, v \notin W$,  none of $(u, v)$ or $(v, u)$ is a bad ancestry pair. In particular, 
 $v \prec_N u$ if and only if $v \prec_{N'} u$. 
Since this holds for all distinct $u,v\in V(N)\setminus W$, and $v\preceq_N v$ if and only
if $v\preceq_{N'} v$ as $V(N)=V(N')$,  we obtain
   $v\preceq_N u$
if and only if
   $ v\preceq_{N'} u$.
 It follows that $W$ is a bypassing set for $N$ and $N'$.
\end{proof}

We now consider also undirected graphs $G=(V,E)$ with vertex set $V(G)\coloneqq V$ and 
where the edge set $E(G)$ is a subset of 2-elementary subsets of $V$. 
Note that, for technical reasons, we allow ``empty'' undirected graphs $(\emptyset,\emptyset)$.
In undirected graphs, edges are unordered pairs and we write $uv$ as shorthand
for the edge $\{u,v\}$.

Given two DC-similar networks $N$ and $N'$ with canonically identified
vertex sets, we define the \emph{bad ancestry graph} $B_{N,N'}$ as the
undirected graph as follows:
\begin{itemize}
	\item 
	$V(B_{N, N'}) = \{ v \in V(N) \colon $ there exists $u \in V(N)$ such that $(u, v)$ or $(v, u)$ is a bad ancestry pair for $N$ and $N'\}$.

	\item 
	$E(B_{N, N'}) = \{uv \colon (u, v) $ or $(v, u)$ is a bad ancestry pair for $N$ and $N'\}$.
\end{itemize}
Thus, $B_{N,N'}$ is either the empty graph or has at least two vertices. In particular, it  has no isolated vertices by definition.

We can now reduce the computation of $\dist^-$ on DC-similar networks to \textsc{Vertex Cover}.
Recall that, given an undirected graph $G$, a \emph{vertex cover} of $G$
is a subset $C\subseteq V(G)$ such that, for every edge
$uv\in E(G)$, at least one of $u$ and $v$ belongs to $C$. The
\textsc{Vertex Cover} problem asks for a vertex cover of minimum size.

The following result is almost immediate from
Lemma~\ref{lem:bypass-at-least-one}: bypassing sets are precisely vertex
covers of the bad ancestry graph. The only minor point is that a bypassing
set is required to be contained in $V^0$, whereas a vertex cover of
$B_{N,N'}$ is a priori just a subset of $V(B_{N,N'})$.

\begin{lemma}\label{lm:bypass-is-vc}
	Let $N, N'$ be DC-similar networks  with canonically identified vertex sets. Then we have $V(B_{N,N'}) \subseteq V^0(N)$. 
	Moreover,	 $W$ is a bypassing set for $N$ and $N'$ if and only if $W$ is a vertex-cover of $B_{N, N'}$. 
\end{lemma}
\begin{proof}
We first show that $V(B_{N,N'}) \subseteq V^0(N)$. 
To this end, we show
that neither the root nor any leaf is  part of a bad ancestry pair. 

Let $r$ be the root in $N$ and $r'$ be the root of $N'$. 
Since $N$ and $N'$ are distinct-cluster networks on the same leaf set, 
it holds that $\CC_N(r)=X=\CC_{N'}(r)$. 
Since $N$ and $N'$ have canonically identified vertex, it follows that
$r=r'$. Since $r$ is the root of both $N$ and $N'$, $r$ is an ancestor of every vertex in both $N$ and $N'$, and thus $r$ is in no
bad ancestry pair. By definition, $r \notin V(B_{N, N'})$.

Now let $x\in X$ be a leaf. For every vertex $u$, we have $ x\preceq_N u \iff x\in \CC_N(u) \iff
x\in \CC_{N'}(u) \iff x\preceq_{N'} u, $
where we used that $N$ and $N'$ are DC-similar with canonically identified
vertex sets. Moreover, since $x$ is a leaf in both networks, no distinct
vertex lies below $x$ in either network. Hence $x$ is not contained in any
bad ancestry pair. Therefore $V(B_{N,N'})\subseteq V^0(N)$.

We now prove the stated equivalence. By
Lemma~\ref{lem:bypass-at-least-one}, $W$ is a bypassing set for $N$ and $N'$ if and only if $W$
contains at least one vertex of every bad ancestry pair. By definition of $B_{N,N'}$, this is
equivalent to saying that $W$ contains at least one endpoint of every edge of $B_{N,N'}$, that is,
that $W\cap V(B_{N,N'})$ is a vertex cover of $B_{N,N'}$. 

Conversely, let $W$ be a vertex cover of $B_{N, N'}$. Since $V(B_{N,N'}) \subseteq V^0(N)$, we have
$W \subseteq V^0(N)$. Moreover, since edges of $B_{N, N'}$ are in 1-to-1 correspondence with bad
ancestry pairs, $W$ contains at least one vertex of each bad ancestry pair $(u, v)$. By
Lemma~\ref{lem:bypass-at-least-one} it is a bypassing set for $N$ and $N'$.  
\end{proof}

A direct consequence of Lemma~\ref{lm:bypass-is-vc} is that the minimum
size of a bypassing set for $N$ and $N'$ is equal to the minimum size of a
vertex cover of $B_{N,N'}$, see Figure~\ref{fig:DC-VC} for an illustrative
example.

\begin{theorem}\label{thm:2VC=dist-}
Let $N, N'$ be DC-similar networks  with canonically identified vertex sets. Then, 
\[
    \dist^-(N,N') = 2\tau(B_{N,N'}),
\]
where $\tau(B_{N,N'})$ denotes the size of a minimum vertex-cover of $B_{N,N'}$.
\end{theorem}

\begin{figure}
	\centering
	\includegraphics[width = 0.9\textwidth]{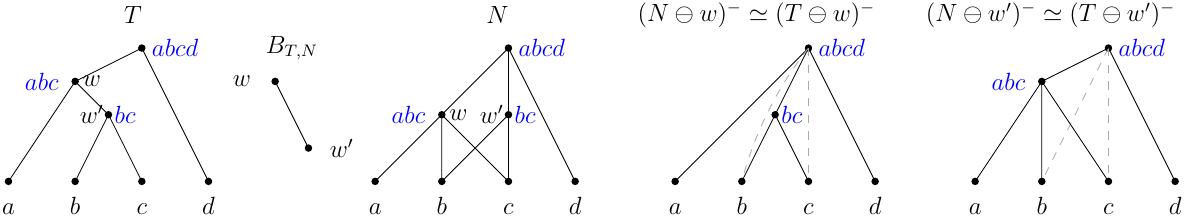}
	\caption{A phylogenetic tree $T$ and a network $N$. Here, 
	$T$ and $N$ are DC-similar	and have canonically identified vertex sets.
	In this example, $(w,w')$ is the unique bad ancestry pair of $T$ and $N$. 
	The bad ancestry graph $B_{T,N}$ consists of the vertices $w$ and $w'$
	together with the undirected edge $\{w,w'\}$. An optimal vertex
	cover of $B_{T,N}$ is either $C= \{w\}$ or $C'= \{w'\}$. 
	Thus, $\tau(B_{T,N})= 1$ and, 
	by Theorem~\ref{thm:2VC=dist-}, $\dist^-(T,N) = 2$.
	In fact, for either choice of the optimal vertex cover, we obtain, by
	Lemma~\ref{lm:bypass-is-vc},
	an optimal $\ominus$-reduction up to shortcut-removal, that is, 
	$(N\ominus w)^- \simeq (T\ominus w)^-$ and 
		$(N\ominus w')^- \simeq (T\ominus w')^-$.
		Dashed arcs indicate shortcuts in the networks
		$N\ominus w $ and $N\ominus w'$.}
	
\label{fig:DC-VC}
\end{figure}

Thus the computation of $\dist^-$ reduces to \textsc{Vertex Cover}. This
does not immediately yield a polynomial-time algorithm, since
\textsc{Vertex Cover} is NP-hard \cite{GareyJohnson1979}. However,
\textsc{Vertex Cover} is fixed-parameter tractable with respect to the size
$k$ of the vertex cover. 
In particular, it can be solved in time
$O(1.25284^k \cdot n^{O(1)}))$, 
with $n$ being the number of
vertices of the input graph \cite{HarrisNarayanaswamy2024}. This improves
the earlier $O(1.2738^k + kn)$ bound of Chen, Kanj and Xia
\cite{ChenKanjXia2010}.

\begin{theorem}\label{thm:dist-minus-FPT}
Let $N$ and $N'$ be DC-similar networks with canonically identified vertex
sets. Put $d=\dist^-(N,N')$ and $n=|V(N)|$. Then computing $\dist^-(N,N')$ is
fixed-parameter tractable with respect to $d$ and
can, in particular, be
computed in time $O(1.11931^d \cdot n^{O(1)})$. 
\end{theorem}
\begin{proof}
Construct the bad ancestry graph $B_{N,N'}$ by comparing the ancestor
relations of all pairs of vertices in $N$ and $N'$. This can be done in
polynomial time. By Lemma~\ref{lm:bypass-is-vc}, minimum bypassing sets
for $N$ and $N'$ correspond to minimum vertex covers of $B_{N,N'}$.
Moreover, by Corollary~\ref{cor:min-bypassing-set},
$ \dist^-(N,N') = 2\tau(B_{N,N'})$
where $\tau(B_{N,N'})$ denotes the minimum size of a vertex cover of
$B_{N,N'}$. Hence it suffices to compute a minimum vertex cover of
$B_{N,N'}$. Applying the algorithm of \cite{HarrisNarayanaswamy2024} with
$k=\tau(B_{N,N'})=d/2$ gives a running time of $O(1.25284^{\frac{d}{2}} \cdot  n^{O(1)}) 
=O(1.11931^{d} \cdot  n^{O(1)})$, which is the claimed running time.
\end{proof}

Note that Theorem~\ref{thm:dist-minus-FPT} contrasts sharply with the complexity of $\dist$. Indeed,
unless $\mathrm{FPT}=\mathrm{W[2]}$, the W[2]-hardness result above rules out an algorithm for
computing $\dist$ with running time $ f(d)\cdot n^{O(1)}$ where $d=\dist(N,N')$, $n$ is the input
size, and $f$ is an arbitrary computable function.

Approximation algorithms for \textsc{Vertex Cover} transfer directly to
$\dist^-$, since $\dist^-(N,N')=2\tau(B_{N,N'})$ for DC-similar networks.
\begin{corollary}\label{cor:dist-minus-approx}
Let $N$ and $N'$ be DC-similar networks with canonically identified vertex sets. If \textsc{Vertex
Cover} admits a polynomial-time $\alpha$-approximation on $B_{N,N'}$, then $\dist^-(N,N')$ admits a
polynomial-time $\alpha$-approximation.
In particular, $\dist^-(N,N')$ admits a polynomial-time $2$-approximation.
\end{corollary}

\begin{proof}
Let $C$ be an $\alpha$-approximate vertex cover of $B_{N,N'}$, that is,
$|C|\leq \alpha\cdot \tau(B_{N,N'})$. By Lemma~\ref{lm:bypass-is-vc}, $C$ is a bypassing set for
$N$ and $N'$ and $C\subseteq V^0(N)$ Hence, by Proposition~\ref{pr:bpset},
$(N\ominus C)^-\simeq (N'\ominus C)^-$. The latter two arguments 
together with Theorem~\ref{thm:dist-minus-FPT} 
imply that $(C,C)\in\mathtt W^-(N,N')$ and that 
\[\dist^-(N,N') \leq 2|C|  \leq 2\alpha\cdot \tau(B_{N,N'}) = \alpha\cdot \dist^-(N,N').\]
Therefore $2|C|$ is an $\alpha$-approximation of $\dist^-(N,N')$.
The final statement follows from the standard polynomial-time $2$-approximation for
\textsc{Vertex Cover}, obtained by computing a maximal matching and taking all endpoints of its
edges, cf.~\cite[Thm~31.1]{Cormen2009}. 
\end{proof}

The standard maximal-matching algorithm for \textsc{Vertex Cover} gives a
polynomial-time $2$-approximation. Slightly better asymptotic approximation
ratios are known for general graphs; in particular, Karakostas~\cite{Kara:09}
gave a polynomial-time approximation ratio of $2-\Theta(1/\sqrt{\log n})$.

This contrasts with the situation for $\dist$. The set-cover reduction used in the proof of
Theorem~\ref{thm:ominus-hard-similar-normal} is approximation-preserving up to the constant factor
$2$. Indeed, every set cover $\mathcal S^\ast$ of size $k$ yields feasible deletion sets of total size
$2k$, while every feasible pair of deletion sets of total size at most $2k$ yields a set cover of
size at most $k$. Thus a polynomial-time constant-factor approximation for $\dist$ would imply a
polynomial-time constant-factor approximation for \textsc{Set Cover}. Since \textsc{Set Cover}
admits no polynomial-time constant-factor approximation unless $\mathrm{P}=\mathrm{NP}$ \cite{dinur2014analytical}, 
the same holds for $\dist$. 
This gives another algorithmic distinction
between $\dist$ and $\dist^-$.

Moreover \textsc{Vertex Cover} is one of the most studied NP-hard
problems, and modern exact solvers perform well on many large practical
instances; see, for example,
\cite{AkibaIwata2016,HespeLammSchulzStrash2019}. Since the bad ancestry
graph $B_{N,N'}$ can be constructed from $N$ and $N'$ in polynomial time,
such solvers can be applied directly to $B_{N,N'}$ in order to determine
$d=\dist^-(N,N')$.

Finally, Proposition~\ref{prop:dc-to-dc-similar} allows us to extend these algorithmic results from
DC-similar networks to arbitrary distinct-cluster networks.

\begin{theorem}\label{thm:distinct-cluster-FPTapprox}
    Let $N$ and $N'$ be distinct-cluster networks on $X$.  Let $ D_N \coloneqq \{v\in V^0(N)\mid \CC_N(v)\notin
\mathfrak C_{N'}\} $ and $ D_{N'}\coloneqq\{v'\in V^0(N')\mid \CC_{N'}(v')\notin \mathfrak C_N\}$.  
    Put $\tilde{d} = \dist^-(N \ominus D_N, N' \ominus D_{N'})$.  

	Then computing $\dist^-(N,N')$ is fixed-parameter tractable with respect to $\tilde d$. In
	particular, it can be computed in time $ O(1.11931^{\tilde d}\cdot n^{O(1)}), $ where $
	n=|V(N)|+|V(N')|. $ Moreover, $\dist^-(N,N')$ admits a polynomial-time $2$-approximation.
\end{theorem}

\begin{proof}
  Let $N$ and $N'$ be distinct-cluster networks on $X$. The cluster sets $\mathfrak C_N$ and
  $\mathfrak C_{N'}$, and hence the sets $D_N$ and $D_{N'}$, can be computed in polynomial time. The
  networks $N\ominus D_N$ and $N'\ominus D_{N'}$ can likewise be constructed in polynomial time. By
  Lemma~\ref{lem:dc-to-dc-similar}, the networks $N\ominus D_N$ and $N'\ominus D_{N'}$ are
  DC-similar. Their canonical identification can be obtained in polynomial time by matching the
  unique vertices having the same clusters.
	
	By Proposition~\ref{prop:dc-to-dc-similar}, $\dist^-(N, N') = |D_N| + |D_{N'}| + \dist^-(N \ominus
	D_N, N' \ominus D_{N'})$. By the definition of $\tilde d$, this becomes $ \dist^-(N,N') =
	|D_N|+|D_{N'}|+\tilde d$. Theorem~\ref{thm:dist-minus-FPT} implies that $\tilde d$ can be
	computed in time $ O(1.11931^{\tilde d}\cdot n^{O(1)}). $ Since $|D_N|+|D_{N'}|$ can be
	computed in polynomial time, the same asymptotic running-time bound applies to the computation of
	$\dist^-(N,N')$.

	It remains to prove the approximation statement. By Corollary~\ref{cor:dist-minus-approx}, there
	is a polynomial-time $2$-approximation for $ \tilde d$
	Let $d_{\mathrm{APP}}$ denote the value returned by this approximation
	algorithm. Then,	$\tilde d	\leq	d_{\mathrm{APP}}	\leq	2\tilde d$. 
	We return $	d^*_{\mathrm{APP}}	\coloneqq	|D_N|+|D_{N'}|+d_{\mathrm{APP}}$.
	Using Proposition~\ref{prop:dc-to-dc-similar}, we obtain	
	$\dist^-(N,N')=|D_N|+|D_{N'}|+\tilde d\leq d^*_{\mathrm{APP}}$.
Moreover,
\[ d^*_{\mathrm{APP}} = |D_N|+|D_{N'}|+d_{\mathrm{APP}} \leq |D_N|+|D_{N'}|+2\tilde d \leq 2(|D_N|+|D_{N'}|+\tilde d)= 2\dist^-(N,N'). \]
Hence $d^*_{\mathrm{APP}}$ is a polynomial-time $2$-approximation of
$\dist^-(N,N')$.
\end{proof}

Finally, the reduction to \textsc{Vertex Cover} yields a direct integer linear programming approach
for computing $\dist^-$ on distinct-cluster networks, using the standard minimum-vertex-cover ILP
formulation, see e.g.~\cite{Hochbaum1997VertexCover}. For DC-similar networks $N$ and $N'$, this
formulation can be applied directly to the bad ancestry graph $B_{N,N'}$. If
$\operatorname{OPT}_{\mathrm{ILP}}$ denotes its optimal value, then $
\dist^-(N,N')=2\operatorname{OPT}_{\mathrm{ILP}}$. For arbitrary distinct-cluster networks, let $D_N$
and $D_{N'}$ be as in Theorem~\ref{thm:distinct-cluster-FPTapprox}. Then, $
N\ominus D_N $ and $ N'\ominus D_{N'} $ 
are DC-similar, applying the same ILP to $B_{N\ominus D_{N},N'\ominus D_{N'}}$ yields $ \dist^-(N,N') =
|D_N|+|D_{N'}| + 2\operatorname{OPT}_{\mathrm{ILP}}$. An ILP formulation for $\dist$ also appears
possible, although it is considerably more involved. In addition to variables encoding the deleted
and identified vertices, it requires auxiliary variables describing the arcs created by the
$\ominus$-operations. Indeed, an arc $u\to v$ occurs in $N\ominus W$ precisely when $N$ contains a
directed $uv$-path whose internal vertices all belong to $W$. Developing such a formulation is
beyond the scope of this work.


\section{NP-hardness result for $\dist^-$}
\label{sec:NPHard-dist-minus}

The preceding results show that, for DC-similar networks $N$ and $N'$,
the complexity of computing $\dist^-(N,N')$ is governed by the structure
of the bad ancestry graph $B_{N,N'}$. In particular, if
$B_{N,N'}$ belongs to a graph class on which \textsc{Vertex Cover} can be
solved in polynomial time, then $\dist^-(N,N')$ can be computed in
polynomial time as well. This applies, for instance, whenever
$B_{N,N'}$ is bipartite, or more generally perfect.

This naturally raises the question which graphs can occur as bad ancestry graphs. 
Equivalently, for which graphs $G$ do there exist networks $N$ and $N'$ such that
$G=B_{N,N'}$? We provide a sufficient
condition for a graph to be realizable as a bad ancestry graph. To this end, we need the following
notation.

An \emph{orientation} of an undirected graph $G=(V,E)$ is a directed graph
$\overrightarrow{G}=(V,A)$ obtained by replacing each edge $uv\in E$ by
exactly one of the two arcs $(u,v)$ or $(v,u)$.
We say that an undirected graph $G$ is \emph{shortcut-free orientable} if there exists an orientation $\overrightarrow{G}$ of $G$ that is a shortcut-free DAG.
\begin{proposition}\label{prop:sf-orientable-to-bnn}
Let $G$ be a connected undirected graph, and suppose that a shortcut-free
orientation $\overrightarrow{G}$ of $G$ is given. Then one can construct,
in polynomial time from $\overrightarrow{G}$, DC-similar shortcut-free
networks $N$ and $N'$ such that
$   G = B_{N,N'}$.
\end{proposition}

\begin{proof}
Let $\overrightarrow{G}$ be a shortcut-free orientation of the connected undirected $G$.  
For $u, v \in V(\overrightarrow{G})$, we denote by $d_{\overrightarrow{G}}(u, v)$ the number of arcs on a shortest directed path from $u$ to $v$ in $\overrightarrow{G}$, and define $d_{\overrightarrow{G}}(u, v) = -1$ if no such path exits.

We construct two auxiliary networks $N$ and $N'$ from
$\overrightarrow{G}$ as follows. During the construction, these networks
may contain shortcuts; these shortcuts will be removed at the end of the
proof. 
We start with $V(N) = V(N') = V(\overrightarrow{G})$.  Then, we apply the following steps.
\begin{enumerate}[noitemsep]
\item 
For each $u,v\in V(\overrightarrow G)$ with
$d_{\overrightarrow G}(u,v)\geq 1$, add the arc $(u,v)$ to $N$.

\item 
For each $u,v\in V(\overrightarrow G)$ with
$d_{\overrightarrow G}(u,v)\geq 2$, add the arc $(u,v)$ to $N'$.

\item 
For each vertex $u \in V(\overrightarrow{G})$, 
add a new leaf $\ell_u$ to both $N$ and $N'$.
Then, add the arc $(u, \ell_u)$ to both $N$ and $N'$. 
Moreover, for each $v \in V(\overrightarrow{G})$ such that $v \prec_{\overrightarrow{G}} u$, add the arc $(u, \ell_v)$ to both $N$ and $N'$.  

\item
Add a new vertex $r$ to both $N$ and $N'$,  which will be the root. 
Also add a new leaf $\ell_r$ and the arc $(r,\ell_r)$ to both networks. Then, in
each of the two networks, add an arc from $r$ to every vertex different
from $r$ that currently has indegree zero.

\item 
Add a new leaf $\ell^*$ to both $N$ and $N'$. 
For each $u \in V(\overrightarrow{G})$, add the arc $(u, \ell^*)$ to both $N$ and $N'$.  
\end{enumerate}

The construction can clearly be carried out in polynomial time.

We first observe that $N$ and $N'$ are networks on the same leaf set. Since
$\overrightarrow G$ is a DAG, its vertices admit a topological ordering.
All arcs added in Steps~(1) and~(2) go from an earlier vertex to a later
vertex with respect to this ordering. The arcs added in Steps~(3) and~(5)
end in leaves, and the arcs added in Step~(4) start at the new vertex
$r$. Hence no directed cycle is created, and both $N$ and $N'$ are DAGs.
Moreover, by Step~(4), every vertex different from $r$ has indegree at
least one, while $r$ has indegree zero. Thus $r$ is the unique root of both
networks. By Step~(5), every vertex $u\in V(\overrightarrow G)$ has the
child $\ell^*$ and is therefore not a leaf. Consequently, the leaves of
both $N$ and $N'$ are precisely
$
    X=\{\ell_u\mid u\in V(\overrightarrow G)\}\cup\{\ell_r,\ell^*\}$.
Thus $N$ and $N'$ are networks on the same leaf set $X$.

Notice that for all $u, v \in V(\overrightarrow{G})$ it holds that 
$v \prec_{\overrightarrow{G}} u$ if and only if $v \prec_N u$.
Indeed,  after Step (1) of the construction, 
the restriction of $N$ to
$V(\overrightarrow G)$ contains precisely the arcs $(u,v)$ corresponding to a
directed $uv$-path $u\leadsto v$ in $\overrightarrow G$. The subsequent steps only add leaves
with incoming arcs and a root with outgoing arcs, and therefore cannot
change ancestor relations between vertices of $V(\overrightarrow G)$.
For $N'$, it holds for all $u, v \in V(\overrightarrow{G})$ that 
$    v\prec_{N'} u
\implies 
    v\prec_{\overrightarrow G} u$. 
Indeed, every arc of $N'$ between vertices of $V(\overrightarrow G)$
corresponds to $uv$-path $u\leadsto v$ containing at least two arcs. 
The converse need not hold, since the arcs of $\overrightarrow G$ themselves were omitted
from $N'$ in Step~(2).

Next we show that $N$ and $N'$ are DC-similar.  
By construction, $V(N)=V(N')$: both networks start with vertex set $V(\overrightarrow G)$, and the
same root and leaves are added to both.

We claim that, for every $u\in V(\overrightarrow G)$, 
\begin{equation}
 \CC_N(u)=\CC_{N'}(u) = \mathcal L_u,  \label{eq:claim1}
\end{equation} 
with $\mathcal L_u \coloneqq \{\ell_v\colon v\preceq_{\overrightarrow G} u\}\cup\{\ell^*\}$.
The inclusion  $\mathcal L_u\subseteq \CC_N(u), \CC_{N'}(u)$ 
follows directly
from Step~(3) and Step~(5), since we explicitly add arcs from $u$ to all leaves $\ell_v$ with
$v\preceq_{\overrightarrow G}u$, and also the arc $(u,\ell^*)$.

Conversely, let $\ell\in \CC_N(u)$. Since the only leaves are  $\ell_r$ and $\ell^*$ and 
the vertices $\ell_v$, $v\in V(\overrightarrow G)$ and since $\ell_r$ is reachable only
from $r$, we have either $\ell=\ell^*$ or $\ell=\ell_v$ for some $v\in V(\overrightarrow G)$.
Suppose $\ell=\ell_v$. If $(u,\ell_v)$ is an arc of $N$, then $v\preceq_{\overrightarrow G}u$ by
construction. Otherwise, there is a path in $N$ from $u$ to some parent $w$ of $\ell_v$. Then $w\in
V(\overrightarrow G)$, and the construction of $N$ implies that $w\preceq_{\overrightarrow G}u$ and
$v\preceq_{\overrightarrow G}w$. Hence $v\preceq_{\overrightarrow G}u$. Thus $ \CC_N(u)\subseteq
\mathcal L_u$. This proves the claimed equality in Equation~\ref{eq:claim1} for $N$.
The argument for $N'$ is similar. Indeed, if $\ell_v\in \CC_{N'}(u)$, then either $(u,\ell_v)$ is
an arc of $N'$, in which case $v\preceq_{\overrightarrow G}u$ by construction, or $u$ reaches a
parent $w$ of $\ell_v$ in $N'$. In the latter case, the construction of $N'$ implies
$w\preceq_{\overrightarrow G}u$ and $v\preceq_{\overrightarrow G}w$, and hence
$v\preceq_{\overrightarrow G}u$. Hence, claimed equality in Equation~\ref{eq:claim1} holds for $N'$. 
It follows that all vertices in $V(\overrightarrow G)$ have the same cluster in $N$ and $N'$.

We next argue that $N$ and $N'$ are distinct-cluster. First notice that $r$ is the only vertex of
$N$ and $N'$ whose cluster contains $\ell_r$ and at least one further leaf. Thus no other vertex has
the same cluster as the root. Consider $u,v\in V(\overrightarrow G)$ with $u\neq v$. Since
$\overrightarrow G$ is acyclic, at least one of $ u\not\preceq_{\overrightarrow G} v $ or $
v\not\preceq_{\overrightarrow G} u$ must hold. Without loss of generality, assume
$u\not\preceq_{\overrightarrow G} v$. Then, by Equation~\ref{eq:claim1}, $C_N(u)$ contains $\ell_u$,
whereas $C_N(v)$ does not. Hence $C_N(u)\neq C_N(v)$. Since, by Equation~\ref{eq:claim1}, $
C_{N'}(u)=C_N(u)$ and $C_{N'}(v)=C_N(v), $  we also have $C_{N'}(u)\neq C_{N'}(v)$.
Finally, every vertex of $N$ or $N'$ outside $\{r\}\cup V(\overrightarrow G)$ is a leaf. Distinct
leaves have distinct singleton clusters. Moreover, for every $u\in V(\overrightarrow G)$, we have $
|C_N(u)|=|C_{N'}(u)|\geq 2 $, since the cluster of $u$ contains both $\ell_u$ and $\ell^*$. Thus no
vertex in $V(\overrightarrow G)$ has the same cluster as a leaf. It follows that both $N$ and $N'$
are distinct-cluster. Since corresponding vertices have the same clusters in $N$ and $N'$, the two
networks are DC-similar.

To summarize, $N$ and $N'$ are DC-similar networks on $X$ with $V(N)=V(N')$. 
This together with Equation~\ref{eq:claim1} and the facts that $\CC_N(r)=\CC_{N'}(r)=X$
and  $\CC_N(x)=\CC_{N'}(x)=\{x\}$ for all $x\in X$
 implies that $N$ and $N'$ have canonically identified vertex sets.

We next focus on $B_{N,N'}$. We show that, for all $u,v\in V(\overrightarrow G)$, $(u,v)\in
E(\overrightarrow G)$ if and only if $ (u,v)$  is a bad ancestry pair of $N$ and $N'$.

First, let $(u,v)\in E(\overrightarrow G)$. Then
$v\prec_{\overrightarrow G}u$, and hence $v\prec_N u$. Thus it remains to
show that $v\not\prec_{N'}u$. Suppose, for contradiction, that there is a
directed $uv$-path
$  (w_1,w_2,\ldots,w_k)$ with $u=w_1$ and $v=w_k$ 
in $N'$. Since $u,v\in V(\overrightarrow G)$ and all added
leaves have outdegree zero, no leaf can occur on this path. Moreover, the
root $r$ cannot occur on this path. Hence all vertices $w_i$ belong to
$V(\overrightarrow G)$.
Since $d_{\overrightarrow G}(u,v)=1$, the arc $(u,v)$ was not added to
$N'$ in Step~(2). Thus the path has length at least two. For every
$i\in\{1,\ldots,k-1\}$, the arc $(w_i,w_{i+1})$ of $N'$ implies, by
construction, that there is a directed path from $w_i$ to $w_{i+1}$ in
$\overrightarrow G$ of length at least two. Hence,
$v= w_k \prec_{\overrightarrow G} w_{k-1} \cdots \prec_{\overrightarrow G} w_1=u$.
Therefore $\overrightarrow G$ contains a directed path from $u$ to $v$
different from the single arc $(u,v)$. 
This contradicts the fact that
$\overrightarrow G$ is shortcut-free. Hence $v\not\prec_{N'}u$, and so
$(u,v)$ is a bad ancestry pair of $N$ and $N'$.

Conversely, let $(u,v)$ be a bad ancestry pair of $N$ and $N'$. The root
and the leaves are not contained in any bad ancestry pair, and hence we may
assume that $u,v\in V(\overrightarrow G)$. 
If $(u,v)\in E(\overrightarrow G)$, then we are done. Thus assume $(u,v)\notin E(\overrightarrow G)$. 
We distinguish two cases.
First suppose that 
$v\prec_{\overrightarrow G}u$.  Since $(u,v)\notin E(\overrightarrow G)$, 
 $d_{\overrightarrow G}(u,v)\geq 2$, and therefore $(u,v)$ is an arc
of $N'$ by Step~(2). Hence $v\prec_{N'}u$. Moreover,
$v\prec_{\overrightarrow G}u$ implies $v\prec_Nu$. Thus $N$ and $N'$
agree on the ancestry relation between $u$ and $v$, and so $(u,v)$ is not
bad; a contradiction.
Now suppose that $v\not\prec_{\overrightarrow G}u$. 
Then $v\not\prec_Nu$. We show that also $v\not\prec_{N'}u$. Suppose, for contradiction, that
$v\prec_{N'}u$. Then there is a directed path from $u$ to $v$ in $N'$. As above, such a path cannot
use leaves or the root, and therefore all its vertices belong to $V(\overrightarrow G)$. 
Each arc of
this path corresponds to a directed path in $\overrightarrow G$. Consequently,
$v\prec_{\overrightarrow G}u$, a contradiction. Hence $v\not\prec_{N'}u$. Thus $N$ and $N'$ again
agree on the ancestry relation between $u$ and $v$, contradicting that $(u,v)$ is bad.

It follows that the bad ancestry pairs of $N$ and $N'$ are precisely the arcs of $\overrightarrow
G$. Therefore, for all $u,v\in V(G)$, $ uv\in E(G) $ if and only if $ uv\in E(B_{N,N'}). $
Since $G$ is connected, it has no isolated vertices. Hence every vertex of
$G$ is incident with some edge of $B_{N,N'}$. Moreover,
$B_{N,N'}$ contains only vertices that are contained in bad ancestry pairs,
and these are all in $V(\overrightarrow G)$. Thus
$   V(G)=V(B_{N,N'})$
Consequently,
$   G = B_{N,N'} $.

It remains to remove shortcuts. Since shortcut removal does, by Lemma~\ref{lem:properties-SF-G_NEW}, 
 not change any ancestor relation, a pair $(u,v)$ is a bad ancestry pair of $N$ and $N'$ if
and only if it is a bad ancestry pair of $N^-$ and $(N')^-$. Therefore,
$ G = B_{N^-,(N')^-}$.
In particular, shortcut removal preserves clusters and does not change the
vertex set. Hence $N^-$ and $(N')^-$ remain DC-similar networks on $X$ with canonically identified
vertex sets. By Lemma~\ref{lem:properties-SF-G_NEW}, 
they are shortcut-free. Replacing $N$ and $N'$ by $N^-$ and $(N')^-$ gives
the desired DC-similar shortcut-free networks.
\end{proof}

We can now reduce \textsc{Vertex Cover} to the computation of $\dist^-$. 
By Proposition~\ref{prop:sf-orientable-to-bnn}, it suffices to identify a
class of graphs on which \textsc{Vertex Cover} is NP-hard and whose members
admit shortcut-free orientations.
We use so-called 3-connected cubic graphs of girth greater than $3$. 
The problem \textsc{Vertex Cover} remains
NP-complete on this class \cite{Uehara1996}.  In particular, these undirected graphs are connected,
triangle-free (i.e.,there are no three distinct vertices that are pairwise adjacent), 
and cubic (i.e., every vertex has degree three). 
We now observe that every triangle-free cubic graph admits a shortcut-free orientation.

\begin{lemma}\label{lem:sf-free-orientable} 
Every connected triangle-free cubic graph $G$ is shortcut-free
orientable. Moreover, a shortcut-free orientation of $G$ can be found in polynomial time.
\end{lemma} 
\begin{proof} 
Let $G$ be a connected triangle-free cubic graph. By Brooks' theorem, 
$G$ is $3$-colorable, since $G$ cannot be graph isomorphic to 
a complete graph on $n\geq 3$ vertices by the triangle-free assumption (see~\cite[Chapter 8]{bondy1976graph}).
Moreover, such a $3$-coloring can be found in polynomial time
and yields a partition $ V(G)=X\dot\cup Y\dot\cup Z $ into three color classes,
    where vertices in the same class have the same color. Since the coloring is proper, each of $X$,
    $Y$, and $Z$ is an independent set; equivalently, no edge of $G$ has both endpoints in the same
    color class.
Orient every edge between $X$ and $Y$ from $X$ to $Y$,
every edge between $X$ and $Z$ from $X$ to $Z$, and every edge between $Y$ and $Z$ from $Y$ to $Z$
which results in $\overrightarrow G$.
It is easy to verify that this orientation is acyclic and thus that  $\overrightarrow G$ is a DAG.
It remains to show that this orientation is shortcut-free. Since directed paths can only move from $X$
to $Y$ to $Z$, every directed path has length at most two. Hence any shortcut would have to be an
arc $x\to z$ with $x\in X$ and $z\in Z$, together with a directed path $ x\to y\to z $ for some
$y\in Y$. But then $xy$, $yz$, and $xz$ are edges of $G$, and therefore $x,y,z$ form a triangle,
contradicting the assumption that $G$ is triangle-free. Thus the orientation is shortcut-free.
\end{proof}

\begin{theorem}\label{thm:dist-minus-NP-hard}
Computing $\dist^-$ is NP-hard, even for DC-similar shortcut-free
networks $N$ and $N'$.
\end{theorem}
\begin{proof}
We reduce from \textsc{Vertex Cover} restricted to 3-connected cubic
graphs of girth greater than $3$, which is NP-hard
\cite[Thm~3]{Uehara1996}. Let $(G,k)$ be such an instance. In particular,
$G$ is connected, triangle-free, and cubic.
By Lemma~\ref{lem:sf-free-orientable}, we can find in polynomial time a
shortcut-free orientation $\overrightarrow G$ of $G$. Applying
Proposition~\ref{prop:sf-orientable-to-bnn} to $\overrightarrow G$, we
obtain in polynomial time DC-similar shortcut-free networks $N$ and $N'$
such that $   G= B_{N,N'}$.
Therefore,    $\tau(G)=\tau(B_{N,N'})$,
where $\tau$ denotes the minimum size of a vertex cover.
By Lemma~\ref{lm:bypass-is-vc}, minimum bypassing sets for $N$ and $N'$
correspond to minimum vertex covers of $B_{N,N'}$. Moreover, by
Corollary~\ref{cor:min-bypassing-set},
$
    \dist^-(N,N')=2\tau(B_{N,N'})$.
Consequently,
$ G $ has a vertex cover of size at most  $k$ if and only if 
    $\dist^-(N,N')\leq 2k$.
Thus, deciding whether $\dist^-(N,N')\leq \delta$ is NP-hard, even for
DC-similar shortcut-free networks. Hence computing $\dist^-$ is NP-hard as
well.

\end{proof}

\section{Summary and outlook}
\label{sec:outlook}

We introduced two operational distances, $\dist$ and $\dist^-$, for comparing rooted phylogenetic
networks by means of the $\ominus$-operator. We established their metric properties and related them
to classical cluster-based dissimilarities. While both distances are computationally hard in
general, $\dist^-$ can be computed in polynomial time for several important network classes. For
distinct-cluster networks, its computation reduces to \textsc{Vertex Cover}, which yields
fixed-parameter, approximation, and integer-programming approaches. Beyond providing a numerical
measure of dissimilarity, optimal deletion sets explicitly localize the structural and ancestral
disagreements between the two networks and determine a largest common reduced ancestry structure

Several natural questions remain open. First, call a network $N$
\emph{$c$-distinct-cluster} if every cluster of $N$ has multiplicity at most $c$. The
distinct-cluster case corresponds to $c=1$, whereas increasing $c$ allows progressively more
vertices to induce the same cluster. Since isomorphism of two $c$-distinct-cluster networks is
fixed-parameter tractable with respect to $c$ (by using the FPT algorithm for \emph{colored isomorphism}, parameterized by the maximum size of a color class, see~\cite{Arvind2015-xf}), it is natural to ask whether
computing $\dist^-(N,N')$ is fixed-parameter tractable when parameterized jointly by $c$ and the
distance
\[
d\coloneqq \dist^-(N,N').
\]
A positive answer would extend the algorithmic results for distinct-cluster networks to a
substantially broader class.

A second direction concerns the structure of the bad ancestry graph. For which classes of
distinct-cluster networks does $B_{N,N'}$ belong to a graph class on which
\textsc{Vertex Cover} can be solved in polynomial time? Such a characterization could identify new
classes of phylogenetic networks for which $\dist^-$ is efficiently computable. More generally,
which graphs can arise as bad ancestry graphs? In particular, is every graph without isolated
vertices realizable as the bad ancestry graph of some pair of DC-similar networks?

It would also be interesting to understand the influence of standard structural parameters and
properties of phylogenetic networks. Is computing $\dist^-$ fixed-parameter tractable with respect
to the level of the input networks, that is, the maximum number of reticulation vertices contained
in a biconnected component? This question appears nontrivial because applying the
$\ominus$-operator may increase the level. Other potentially useful parameters include the total
number of reticulation vertices and the maximum indegree. More generally, structural restrictions
on the ancestor relations of certain network classes may impose useful properties on the associated
bad ancestry graphs and thereby lead to more efficient algorithms.

\section{Acknowledgements}

The authors thank Anna Lindeberg for her feedback on a previous version of the manuscript, and for pointing out helpful references. GES would like to thank Université de Sherbrooke, for a fully-funded extended stay during which major progress were made on this project. 
ML acknowledges financial support from the Natural Sciences and Engineering Research Council of Canada (NSERC) Discovery program, and from the Fonds de recherche du Québec – Nature et technologies (FRQNT) NOVA program.

\bibliographystyle{spbasic}
\bibliography{biblio}

\end{document}